\documentclass[a4paper,UKenglish,cleveref, autoref, thm-restate]{lipics-v2021}
\usepackage[T1]{fontenc}
\usepackage{inputenc}
\usepackage{amssymb,amsmath}
\usepackage{mathtools}
\usepackage{graphicx}
\usepackage{color}
\usepackage{todonotes}
\usepackage{xspace}
\usepackage{enumerate}

\hideLIPIcs
\nolinenumbers
\newcommand{\new}[1]{#1}

\renewcommand{\int}{\mathrm{int}}
\renewcommand{\leq}{\leqslant}
\renewcommand{\geq}{\geqslant}
\newcommand{\Line}{\mathrm{line}}
\newcommand{\Reals}{\mathbb{R}}
\newcommand{\conv}{\mathrm{conv}}
\newcommand{\etal}{\emph{et al.}}

\newcommand{\li}[1]{\overline{#1}}

\newcommand{\NP}{\texttt{NP}}
\newcommand{\classP}{\texttt{P}}
\newcommand{\APX}[1]{\texttt{APX}}
\newcommand{\eps}{\varepsilon}
\renewcommand{\angle}{\sphericalangle}

\newcommand{\diam}{\mathrm{diam}}
\newcommand{\dist}{\mathrm{dist}}
\newcommand{\Area}{\mathrm{area}}
\newcommand{\peri}{\mathrm{per}}      % perimeter

\newcommand{\sol}{C}                        % notation for a solution
\newcommand{\opt}{\mbox{{\sc opt}}\xspace}  % value of optimal solution
\newcommand{\myopt}{\mathrm{opt}}           % opt in subscripts
\newcommand{\Popt}{\sol_{\myopt}}              % optimal solution (polytope) itself
\newcommand{\Gs}{G(\sigma)}                % grid inside a square $\sigma$
\newcommand{\Ps}{\sol(\sigma)}                % optimal solution restricted to a grid inside a square $\sigma$
\newcommand{\ch}{\mbox{\sc ch}}            % convex hull

\newcommand{\vbot}{v_{\mathrm{bot}}}
\newcommand{\vbotc}{\overline{v}_{\mathrm{bot}}}
\newcommand{\pbot}{p_{\mathrm{bot}}}
\newcommand{\pbotc}{\overline{p}_{\mathrm{bot}}}
\newcommand{\mylength}{\mathrm{length}}
\newcommand{\seq}{\Xi}
\newcounter{ctr}
\loop
 \stepcounter{ctr}
 \expandafter\edef\csname c\Alph{ctr}\endcsname{\noexpand\mathcal{\Alph{ctr}}}
 \expandafter\edef\csname b\alph{ctr}\endcsname{\noexpand\mathbf{\alph{ctr}}}
 \expandafter\edef\csname b\Alph{ctr}\endcsname{\noexpand\mathbf{\Alph{ctr}}}
\ifnum\thectr<26
\repeat

\renewcommand{\bd}{\partial}

\newenvironment{myquote}%
  {\list{}{\leftmargin=4mm\rightmargin=4mm}\item[]}%
  {\endlist}
\newenvironment{claiminproof}{\begin{myquote}\noindent\emph{Claim.}}{\end{myquote}}
\newenvironment{proofinproof}{\begin{myquote}\noindent\emph{Proof.}}{\hfill $\lhd$ \end{myquote}}

\title{Computing Smallest Convex Intersecting Polygons}
\funding{Mark de Berg is supported by the  Dutch Research Council (NWO) through    Gravitation-grant NETWORKS-024.002.003.}
\author{Antonios Antoniadis}{Department for Applied Mathematics, University of Twente, the Netherlands}{a.antoniadis@utwente.nl}{}{}
\author{Mark de Berg}{Department of Mathematics and Computer Science, TU Eindhoven, the Netherlands}{M.T.d.Berg@tue.nl}{https://orcid.org/0000-0001-5770-3784}{}
\author{S\'andor Kisfaludi-Bak}{Department of Computer Science, Aalto University, Espoo, Finland}{sandor.kisfaludi-bak@aalto.fi}{https://orcid.org/
0000-0002-6856-2902}{}
\author{Antonis Skarlatos}{Department of Computer Science, University of Salzburg, Austria}{antonis.skarlatos@plus.ac.at}{https://orcid.org/
0000-0002-7623-9419}{Part of the work was done during an internship at the Max Planck Institute for Informatics in Saarbr\"ucken, Germany.}
\authorrunning{A. Antoniadis, M. de Berg, S. Kisfaludi-Bak and A. Skarlatos}

\Copyright{Antonios Antoniadis, Mark de Berg, S\'andor Kisfaludi-Bak, Antonis Skarlatos}

\index{Antoniadis, Antonios}
\index{de Berg, Mark}
\index{Kisfaludi-Bak, S\'andor}
\index{Skarlatos, Antonis}

\ccsdesc[100]{Theory of computation~Design and analysis of algorithms}

\keywords{convex hull, imprecise points, computational geometry}

\EventEditors{John Q. Open and Joan R. Access}
\EventNoEds{2}
\EventLongTitle{42nd Conference on Very Important Topics (CVIT 2016)}
\EventShortTitle{CVIT 2016}
\EventAcronym{CVIT}
\EventYear{2016}
\EventDate{December 24--27, 2016}
\EventLocation{Little Whinging, United Kingdom}
\EventLogo{}
\SeriesVolume{42}
\ArticleNo{23}
\begin{document}
\maketitle
%-----------------------------------------------------------------------

%-----------------------------------------------------------------------
\begin{abstract}
A polygon $\sol$ is an \emph{intersecting polygon} for a set $\cO$ of objects in $\Reals^2$
if $\sol$ intersects each object in~$\cO$, \new{where the polygon includes its interior}.
We study the problem of computing the minimum-perimeter intersecting polygon and the minimum-area convex intersecting polygon for a given set $\cO$ of objects. 
We present an FPTAS for both problems for the case where $\cO$ is a set of possibly intersecting convex polygons in the plane of total complexity~$n$.
      
Furthermore, we present an exact polynomial-time algorithm for the minimum-perimeter intersecting polygon for the case where $\cO$ is a set of $n$ possibly intersecting segments in the plane.  So far, polynomial-time exact algorithms were only known for the minimum perimeter intersecting polygon of lines or of disjoint segments.
\end{abstract}
%-----------------------------------------------------------------------

%\input{introduction}
%-----------------------------------------------------------------------
\section{Introduction}
%-----------------------------------------------------------------------
% Computing convex sets under various constraints is a central problem in optimization. 
% In the most basic setting, we are looking for the convex hull of a set of points in the plane: 
% the \emph{smallest} convex set that covers all points. 
% The convex hull is one of the founding problems in computational geometry, where several decades of research effort has produced practical and provably efficient algorithms, see~\cite{jarvis1973identification,preparata1977convex,graham1983finding,kirkpatrick1986ultimate,chan1996optimal} as well as~\cite[Chapter 26]{toth2017handbook}. 
Convex hulls are among the most fundamental objects studied in computational geometry.
In fact, the problem of designing efficient algorithms to compute the convex hull
of a planar point set~$\cO$---the smallest convex set containing~$\cO$---is one of the problems that started 
the field~\cite{jarvis1973identification,preparata1977convex}.
Since the early days, the problem has been studied extensively, resulting in
practical and provably efficient algorithms, in the plane as well as in higher dimensions; 
see the survey by Seidel~\cite[Chapter 26]{toth2017handbook} for an overview.

A natural generalization is to consider convex hulls for a collection~$\cO$ 
of geometric objects (instead of points) in $\Reals^2$. Note that the convex hull of a set
of polygonal objects is the same as the convex hull of the vertices of the objects.
Hence, such convex hulls can be computed using algorithms for computing the convex hull of a point set.
A different generalization, which leads to more challenging algorithmic questions, is
to consider the smallest convex set that~\emph{intersects} all objects in $\cO$. Thus,
instead of requiring the convex set to fully contain each object from~$\cO$, we
only require that it has a non-empty intersection with each object.

Notice that in case of points, the ``smallest'' set is well-defined: if convex sets $C_1$ and $C_2$ 
both contain a point set $\cO$, then $C_1\cap C_2$ also contains~$\cO$. Hence, the convex hull
of a point set $\cO$ can be defined as the intersection of all convex sets containing~$\cO$.
When $\cO$ consists of objects, however, this is no longer true, and the term ``smallest'' is ambiguous. 
In the present paper we consider two variants: given a set $\cO$ of possibly intersecting convex polygons in $\Reals^2$
of total complexity~$n$, find a convex set of minimum perimeter that intersects all objects in~$\cO$, 
or a convex set of minimum area that intersects all objects in~$\cO$.

Observe that a minimum-perimeter connected intersecting set~$C$ for~$\cO$ must be a convex polygon.
To see this, observe that for any object $o\in\cO$ we can select a point~$p_o\in o\cap C$, 
and take the convex hull of these points; the result is a feasible convex polygon whose perimeter 
is no longer than that of $C$. Thus the convexity of the solution could be omitted from
the problem statement. This contrasts with the minimum-area problem, where there is always 
an intersecting polygon of zero area, namely, a tree. The convexity requirement is therefore
essential in the problem statement. Note that it is still true that the minimum-area convex
intersecting set is a polygon: \new{given a convex solution, we can again take the convex hull of the points $p_o$ and get a feasible solution whose area is not greater than the area of the initial convex solution}.
We also remark that the two problems typically have different optima. 
If $\cO$ consists of the three edges of an equilateral triangle, then the minimum-area solution 
is a line segment (that is, a degenerate polygon of zero area), whereas the minimum-perimeter solution 
is the triangle whose vertices are the midpoints of the edges.

The problem of computing minimum-area or minimum-perimeter convex intersecting polygons,
as well as several related problems, have already been studied. 
Dumitrescu and Jiang~\cite{DumitrescuJ12} considered the minimum-perimeter intersecting polygon problem. 
They gave a constant-factor approximation algorithm as well as a PTAS for the case when the objects
in~$\cO$ are segments or convex polygons. They achieved a running time of $n^{O(1)}/\eps+2^{O(1/\eps^{2/3})}n$. 
% In their conclusion, they point to a challenge related to this problem: the difficulty of \NP-hardness proofs because of the convexity of the optimum. 
They also prove that computing a minimum-perimeter intersecting polygon for a set
$\cO$ of non-convex polygons (or polygonal chains) is \NP-hard. For convex input objects,
however, the hardness proof fails. Hence, Dumitrescu and Jiang  ask the following question. 
%-----------------------------------------------------------------------
\begin{quote}
    \textbf{Question 1.} Is the problem of computing a minimum-perimeter intersecting polygon of a set of segments \NP-hard?
\end{quote}
%-----------------------------------------------------------------------
In case of \emph{disjoint} segments,  
a minimum-perimeter intersecting polygon can be found in polynomial time~\cite{javad2010convex,jia2017minimum}, but for intersecting segments the question is still open.

The problem of computing smallest intersecting polygons for a set~$\cO$ of objects has also
been studied in works on \emph{imprecise points}. Now the input is a set of points,
but the the exact locations of the points are unknown. Instead, for each point one is given 
a region where the point can lie. One can then ask questions such as: what is the largest possible
convex hull of the imprecise points? And what is the smallest possible convex hull?
If we consider the objects in our input set~$\cO$ as the regions for the imprecise points, then the latter
question is exactly the same as our problem of finding smallest intersecting convex sets.
In this setup both the minimum-perimeter and minimum-area problem have been considered,
for sets $\cO$ consisting of convex regions of total complexity~$n$. 
There are exact polynomial-time algorithms for minimum (and maximum) perimeter and area,
for the special case where $\cO$ consists of horizontal line segments or axis-parallel squares~\cite{LofflerK10}. 
Surprisingly, some of these problems are \NP-hard, such as the \emph{maximum}-area/perimeter problems for segments. 
This gave rise to the study of approximation algorithms and approximation schemes~\cite{KreveldL08}. 
% There are known polynomial time approximations schemes (PTASes) for maximum area/perimeter with disks and $k$-gons~\cite{KreveldL08}.

In some cases, the minimum-perimeter problem can be phrased as a travelling salesman problem with neighborhoods (TSPN). 
Here the goal is to find the shortest closed curve intersecting all objects from the given set~$\cO$. 
In general, an optimal TSPN tour need not be convex, but one can show that 
in \new{the} case of lines or rays, an optimal tour is always convex: if a convex polygon
intersects a line (or a ray) then its boundary intersects the line (resp.~the ray).
Therefore, computing a minimum-perimeter intersecting polygon of lines (or rays) is the same problem as TSPN with line neighborhoods (resp.~ray neighborhoods). 
TSPN of lines in $\Reals^2$ admits a polynomial-time algorithm~\cite{tour_seq_polyg}. 
In higher dimensions, TSPN has a PTAS for hyperplane neighborhoods~\cite{AntoniadisFHS20}, 
but notice that this is not the natural generalization of the minimum-intersecting polygon
problem. Tan~\cite{tan_tour_rays} proposed an exact algorithm for TSPN of rays in~$\Reals^2$, 
but there seems to be an error in the argument; see Appendix~\ref{app:counterexample}. 
At the time of writing this article, we believe that a polynomial-time algorithm for TSPN of rays is not known, 
but there is a \new{constant-factor approximation} algorithm
due to Dumitrescu~\cite{dumitrescu2012approx}, as well as a PTAS~\cite{DumitrescuJ12}.

%-----------------------------------------------------------------------
\subparagraph*{Our results.}
%-----------------------------------------------------------------------
In order to resolve Question 1, we first need to establish a good structural
understanding and a dynamic programming algorithm. It turns out that the algorithm
can also be used for approximation. We give dynamic-programming-based approximation 
schemes for the minimum-perimeter and 
minimum-area convex intersecting polygon problems. 
Our first algorithm is a fully polynomial time approximation scheme (FPTAS) 
for the minimum-perimeter problem of arbitrary convex objects of total complexity~$n$.
%-----------------------------------------------------------------------------------
\begin{theorem}\label{thm:2d-fptas-perimeter}
Let $\cO$ be a set of convex polygons of total complexity $n$ in $\Reals^2$ and let $\opt$ be the
minimum perimeter of an intersecting convex polygon for $\cO$. For any given $\eps>0$,
we can compute an intersecting polygon for $\cO$ whose perimeter is at most $(1+\eps)\cdot\opt$,
in $O(n^{2.373}/\eps + n/\eps^8)$ time.
\end{theorem}
%-----------------------------------------------------------------------------------
This is a vast improvement over the PTAS given by Dumitrescu and Jiang~\cite{DumitrescuJ12}, 
as the dependence on $1/\eps$ is only polynomial in our algorithm. Our approximation algorithms work 
in a word-RAM model, where input polygons are defined by the coordinates of their vertices, 
and where each coordinate is a word of $O(\log n)$ bits.

We also get a similar approximation scheme for the minimum area problem, albeit with a slower running time. 
Here we rely more strongly on the fact that the objects of $\cO$ are convex \emph{polygons}, 
and an extension to (for example) disks is an interesting open question.  The minimum-perimeter FPTAS needs to be adapted to the minimum-area setting.
%-----------------------------------------------------------------------------------
\begin{theorem}\label{thm:2d-fptas-area}
Let $\cO$ be a set of convex polygons of total complexity $n$ in $\Reals^2$ and let $\opt$ be the
minimum area of an intersecting convex polygon for $\cO$. For any given $\eps>0$,
we can compute a convex intersecting polygon for $\cO$ whose area is at most $(1+\eps)\cdot\opt$,
in $O(n^{17}\log(1/\eps) + n^{11}/\eps^{24})$ time.
\end{theorem}
%-----------------------------------------------------------------------------------
We remark that both Theorem~\ref{thm:2d-fptas-perimeter} and Theorem~\ref{thm:2d-fptas-area} work if the input has \emph{polytopes} instead of polygons, that is, when each object is the intersection of some half-planes.

While the dynamic programming algorithm developed above is crucial to get an exact algorithm, we are still several steps from being able to resolve Question~1.
The main challenge here is that the vertices of the optimum intersecting polygon can be located 
at arbitrary boundary points in $\cO$, and there is no known way to discretize the problem. 
We introduce a subroutine that uses an algorithm of Dror~\etal~\cite{tour_seq_polyg} 
to compute parts of the minimum-perimeter intersecting polygon that contain no vertices of input objects. 
We are able to achieve a polynomial-time algorithm (on a real-RAM machine) for the minimum perimeter 
intersecting polygon problem only when the objects are line segments.
%-----------------------------------------------------------------------------------
\begin{theorem}\label{thm:2d-exact}
Let $\cO$ be a set of $n$ line segments in the plane. Then we can compute a 
minimum-perimeter intersecting polygon for $\cO$ in $O(n^{9}\log n)$ time. 
\end{theorem}
%-----------------------------------------------------------------------------------
If $\classP\neq \NP$, then this gives a direct negative answer to Question 1. 
The theorem also extends to the case of rays (this is the scenario studied by Tan~\cite{tan_tour_rays}; see Appendix~\ref{app:counterexample}.
%-----------------------------------------------------------------------------------
\subparagraph*{Our techniques.}
%-----------------------------------------------------------------------------------
Our approximation algorithms both compute an approximate solution whose vertices are 
from some fine grid. To determine a suitable grid resolution, we need to be able to compute 
lower bounds on $\opt$, which is non-trivial. It is also non-trivial to know
where to place the grid, such that it is guaranteed to contain an approximate solution. 
The problem is that our lower bound gives us the location of a solution that is a constant-factor
approximation, but this can be far from the location of a $(1+\eps)$-approximation.
Hence, for the minimum-area problem we generate a collection of grids, one of which is 
guaranteed to contain a $(1+\eps)$-approximate solution. Finally, we face some further difficulties 
since a square grid may be insufficient: the optimum intersecting polygon may be extremely 
(exponentially) thin and long, and of area close to zero. In such cases there is no square grid of 
polynomial size that would contain a good solution. These problems are resolved in Section~\ref{sec:location}.

Section~\ref{sec:FPTAS} presents our dynamic programming algorithm for minimum perimeter. In the dynamic programming the main technical difficulty lies in the fact that it is not clear what subset of objects should be visited in each subproblem.
A portion of the optimum's boundary could in principle be tasked with intersecting an arbitrary subset of $\cO$, while some of the objects in $\cO$ need not be intersected by the optimum boundary and will simply be covered by the interior of the optimum intersecting polygon: a na\"ive approach therefore would not yield a polynomial-time algorithm. Our carefully designed subproblems have a clear corresponding set of objects to ``visit'', using orderings of certain tangents of input objects for this purpose. The minimum area problem uses a similar dynamic program, see Section~\ref{sec:area-fptas}
for its details.

Finally, in order to present our exact algorithm in Section~\ref{sec:exact}, we need to modify our dynamic program to deal with subproblems where the vertices of a convex chain do not come from a discretized set. In such cases, we have to find the order in which the objects of $\cO$ are visited by the chain. We are able to prove a specific ordering only in the case when the objects are line segments. The order then allows us to invoke the algorithm of Dror~\etal~\cite{tour_seq_polyg} in a black-box manner.

%\input{locating-opt}

%-------------------------------------------------------------------------------
\section{Locating an optimal solution}
\label{sec:location}
%-------------------------------------------------------------------------------
The algorithms to be presented in subsequent sections need to approximately know 
the size and location of a smallest intersecting polygon. 
We use an algorithm from~\cite{DumitrescuJ12} to locate the minimum-perimeter intersecting polygon. With respect to the minimum-area intersecting polygon we prove that either there is a solution with a constant number of vertices (that can be computed with a different algorithm), or it is sufficient to consider polygons whose vertices are from a grid which comes from a polynomial collection of different grids.
%In this section we give a short overview of the relevant results of \cite{DumitrescuJ12} and also show how the optimum can be ``located''  for the minimum-perimeter and minimum-area problem.

%-------------------------------------------------------------------------------
\subparagraph*{Locating the minimum-perimeter optimum.}
%-------------------------------------------------------------------------------
%-------------------------------------------------------------------------------
For the minimum-perimeter intersecting polygon of a set $\cO$ of convex objects, Dumitrescu and Jiang~\cite{DumitrescuJ12} present an algorithm~$A1$ that, for
a given~$\eps_1>0$, outputs a rectangle $R$ intersecting all input objects $\cO$ and with perimeter at most $\frac{4}{\pi}(1+\eps_1)\opt$. At a high level, $A1$ guesses an orientation of the rectangle among $\lceil\frac{\pi}{4\eps_1}\rceil$ many discrete orientations and then uses a linear program to identify the smallest perimeter rectangle of that orientation that intersects $\cO$. In~\cite{DumitrescuJ12} it is described how Algorithm~$A1$ is used to locate an optimal solution if the input objects are convex polygons. In particular, for any $\eps>0$ running $A1$ with $\eps_1=\frac{\eps}{2+\eps}$ gives a rectangle $R$. %Appendix~\ref{sec:extend-nonpoly} discusses how this step can be achieved if $\cO$ has convex objects with constant-degree algebraic curve boundaries.
Let $\sigma$ be the square that is concentric and parallel to $R$ and has a side length of $3\cdot\peri(R)$. Then the following holds.%\footnote{Lemma~$3$ in~\cite{DumitrescuJ12} uses only the convexity of objects in $\cO$ and thus holds in our more general setting.}
%-------------------------------------------------------------------------------
\begin{lemma}[Lemma~$3$ in~\cite{DumitrescuJ12}]
Suppose that $\peri(R)\ge (1+\eps)\opt$. Then there is an optimum polygon $\Popt$ that is covered by $\sigma$.
\label{lem:dum_locate}
\end{lemma}
%-------------------------------------------------------------------------------
Algorithm $A1$ needs to solve $O(1/\eps)$ many linear programs with $O(n)$ variables and $O(n)$ constraints each. 
Thus $R$ and $\sigma$ can be found in $O(T_{\mathrm{LP}}(n)/\eps)$ time, where $O(T_{\mathrm{LP}}(n)$ is the running time of an LP solver with $O(n)$ variables and  $O(n)$ constraints.  The state-of-the-art LP solver 
by Jiang et al.~\cite{Jiang0WZ21} achieves a running time better than $O(n^{2.373})$.
Lemma~\ref{lem:dum_locate} directly implies that if $\peri(R)\ge (1+\eps)\peri(\Popt)$, then
\begin{align}
    \label{eq:lb_diam}
    \diam(\Popt) \le \diam(\sigma) = 3\sqrt{2}\peri(R) \le 3\sqrt{2}\frac{4}{\pi}(1+\eps)\peri(\Popt)= O(\diam(\Popt)).
\end{align}
\subparagraph*{The shape and location of the minimum-area optimum.}
%-------------------------------------------------------------------------------
For the rest of this section, let $X$ denote the set of vertices in the planar arrangement given by $\cO$.
%-------------------------------------------------------------------------------
\begin{lemma}\label{lem:minvolumevertexlocation}
Let $\Popt$ be a minimum-area intersecting polygon for the input $\cO$ that has the 
minimum number of vertices, and among such polygons has the maximum number of points from $X$ on its boundary.
Then for any vertex $v$ of $\Popt$ that is not in $X$, the relative interior of at least one side of $\Popt$ adjacent to $v$ contains a point of $X$.
\end{lemma}
%-------------------------------------------------------------------------------
\begin{proof}
Suppose for a contradiction that $v\not \in X$ and that the relative interior of 
the sides in $\Popt$ adjacent to $v$ are disjoint from $X$. Observe that $v$ must be on the boundary of an input object, so it is in the relative interior of an edge $e$ of an input object, see Fig.~\ref{fig:locate_min_area}(i). Then there exists a vector parallel to $e$ along which we can move $v$ while fixing its neighboring vertices in $\Popt$, without \new{increasing} the area of $\Popt$. This movement can be continued until we hit a point in $X$, or the angle of the polygon becomes $\pi$ at $v_1$ or $v_2$. 
% Indeed, if we move further in the former case we might fail to overlap some object of $\cO$, while in the latter case we would violate the convexity requirement. 
As a result, we end up with a feasible polygon $S$ whose area is no greater than that of $\Popt$, and it has one less vertex or at least one more point of $X$ on its boundary. This contradicts the properties of~$\Popt$.
\end{proof}
%-------------------------------------------------------------------------------
\begin{figure}[t]
\centering
\includegraphics{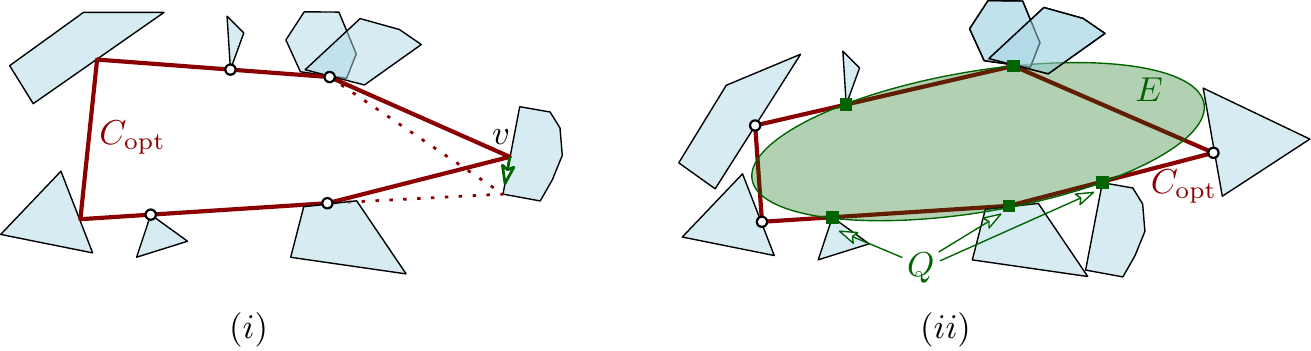}
\caption{(i) If $v$ and adjacent sides of $\Popt$ are disjoint from $X$, then we can slide $v$ without \new{increasing} the area of $\Popt$. (ii) The points of $Q$ (green) and their circumscribed ellipse $E$.}\label{fig:locate_min_area}
\end{figure}
%-------------------------------------------------------------------------------
The following lemma is the key to the success of the algorithm.
%-------------------------------------------------------------------------------
\begin{lemma}\label{lem:constantsize_orflat}
For any given set of input polygons $\cO$ and $0<\eps<1$ there is an intersecting polygon $\sol$ of area $(1+\eps)\opt$ which either has at most $8$ vertices, or its vertices are in a rectangular grid $G$ of size $O(1/\eps^3)\times O(1/\eps^3)$ where $G$ belongs to a collection $\cG$ of grids that can be generated in polynomial time.
\end{lemma}

\begin{proof}
Let $\Popt$ be a minimum-area convex intersecting polygon for the input $\cO$ that has the minimum number of vertices, and among such polygons has the maximum number of points from $X$ on its boundary. Notice that if $\Popt$ has $0$ area, then it is a doubled segment, and thus has only $2$ vertices.

Suppose now that $\Popt$ has at least $9$ vertices, and it has positive area. By Lemma~\ref{lem:minvolumevertexlocation}, $\bd \Popt$ has at least four points from $X$ that are on four distinct edges of $\Popt$. Let $Q=\Popt \cap X$. Since the area of $\Popt$ is not $0$ and  $Q$ contains points on four distinct edges, we have that $Q$ forms a convex polygon of positive area.

Let $E$ be the minimum-area ellipse (of any orientation) covering $Q$. Such an ellipse must contain at most $5$ points from $Q$. Therefore $E$ is uniquely defined by at most $5$ vertices from~$Q$, see Fig.~\ref{fig:locate_min_area}(ii). Consequently, we can generate a set $\cB$ of candidate ellipses in time $O(|X|^5)$ time that is guaranteed to contain $E$. As a consequence of John's theorem~\cite{john1948extremum,ball1992ellipsoids} we can scale $E$ by a factor $1/2$ from the center of $E$ to get an ellipse $E'$ that is contained in $E$.

Next, we apply an affinity that makes $E$ into a circle. Note that the transformation changes the area of all polygons by the same multiplicative constant, therefore the transformation preserves optima and multiplicative approximations of area. We will keep using our previous notations on this transformed instance for the rest of this proof. Without loss of generality assume that $E$ is the unit radius circle centered at the origin. Consequently, $E'$ is the unit diameter circle centered at the origin.

The points of $Q$ subdivide $\bd \Popt$ into $|Q|$ distinct sections. A section is called a \emph{spike} if it contains some vertex of $\Popt$ outside the disk of radius $32/\eps$. If $v$ is a vertex, then let $S(v)$ denote the convex hull of the section of $v$.

Suppose that $\Popt$ has a vertex $v$ that is outside the disk of radius $32/\eps$. We will show that in this case there is a $(1+\eps)$-approximate solution that has at most $8$ vertices. Let $q_v,q'_v\in Q$ be the endpoints of the section of $v$. Notice that $\Popt$ covers the triangle $vq_vq'_v$. Since $E'$ is contained in $\Popt$, we know that the half-cone given by the rays $vq_v$ and $vq'_v$ contains $E'$. Thus the shape given by the touching segments from $v$ to $E'$ and the disk of $E'$ is covered by $\Popt$, see Fig.~\ref{fig:spike}(i). This shape has area area at least $16/\eps$. Consequently, if $\Popt$ has a spike, then 
\begin{equation}\label{eq:spiked_area}
    \Area(\Popt)\geq 16/\eps.
\end{equation}

We now claim that $\Popt$ has at most two sections which have a vertex at distance strictly greater than $2$ from the origin. To prove the claim, suppose the contrary, that $v,w,x$ are vertices in disjoint sections, and each of $v,w,x$ are at distance strictly greater than $2$ from the origin. It follows that each of the three portions of $\bd \Popt$ defined by $v,w,x$ must contain some point of $Q$, and in particular, some point that falls in the unit disk. Since the triangle $T$ given by $v,w,x$ is covered by $\Popt$, we have that each side of $T$ must intersect the unit disk (as otherwise at least one of the three boundary portions would be disjoint from the unit disk). Assume without loss of generality that $T$ has its largest angle at $v$, i.e., it has angle at least $\pi/3$ at $v$. Then it follows that $v$ is at distance at most $2$ from the origin, which is a contradiction. This concludes the proof of our claim.

\begin{figure}[t]
\centering
\includegraphics[width=\textwidth]{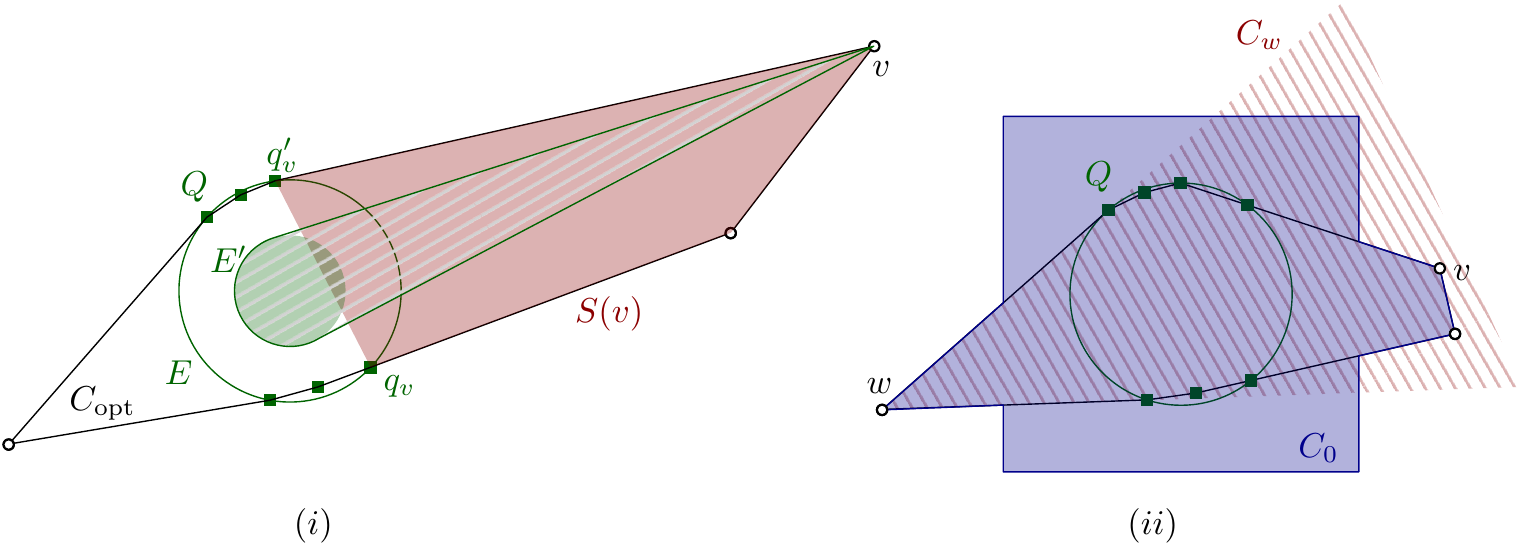}
\caption{(i) A spike at vertex $v$. The section convex hull $S(v)$ is shaded with red. (ii) Building blocks of a constant-size feasible polygon with at most two spikes and the radius $2$ square. The polygon $C_0$ is shaded blue, and the red falling pattern indicates the cone $C_w$.}\label{fig:spike}
\end{figure}

Our claim implies that there are vertices $v,w\in \Popt$ such that $\Popt$ can be covered by the union of the sections $S(v)$, $S(w)$, and the square $[-2,2]\times [-2,2]$. Let $\sol_0=S(v)\cup S(w) \cup [-2,2]\times [-2,2]$. Note that $\sol_0$ is a polygon covering $\Popt$. Consider the (potentially unbounded polygon) $\sol_v$ whose sides are the sides of the spike, where we extend the sides adjacent to $q_v$ and $q'_v$ until they meet (or into rays), see Fig.~\ref{fig:spike}(ii) for an illustration. Notice that $\sol_v$ covers $\Popt$. The polygon $\sol=\sol_0 \cap \sol_v \cap \sol_w$ is convex and covers $\Popt$, and $\Area(\sol)-\Area(\Popt)\leq \Area([-2,2]\times [-2,2])=16$. Thus if $\Popt$ has a spike, then by the bound~\eqref{eq:spiked_area} the polygon $\sol$ has area at most $\Area(\Popt)+16 \leq (1+\eps)\Area(\Popt)$. Observe that $\sol$ has at most $8$ vertices: it has at most $2+2$ vertices on the spikes. Every further vertex is either an original vertex of the square, or it arises after some side of $\sol_v$ or $\sol_w$ incident to $Q$ cuts off at least one vertex, keeping the number of vertices unchanged.

It remains to show that if $|V(\Popt)|\geq 9$ and all vertices of $\Popt$ are within distance $32/\eps$
from the origin, then we can generate a polynomial-size grid set $\cG$ with the desired properties. Recall that we can generate a set $\cB$ of candidate circumscribed ellipses of $Q$ in time $O(|X|^5)=O(n^{10})$ time, that is, for all at most 5-tuples of $X$ we compute the corresponding circumscribed ellipse and the affine transform to make the ellipse into a circle. For each ellipse the computations and the transformation takes $O(1)$ time. We set the circle's radius as unit and fix a coordinate system centered at the circle center. Let $G$ be the square grid that subdivides each side of the square $\sigma=[-32/\eps,32/\eps]\times [-32/\eps,32/\eps]$ into $k=\lfloor 2^{16}/\eps^3\rfloor$ equal segments. The resulting grid $G$ has size $k\times k = O(1/\eps^3)\times O(1/\eps^3)$. Let $\cG$ be the set of $O(n^{10})$ grids generated this way.
For the sake of simplifying the proof and the illustrations, we apply the affine transform to the entire instance, but we need not do that transformation to generate $G$: each ellipse defines a rectangular grid $G$ on the original plane whose points we can access in constant time.

It remains to show that there is a polygon $\sol$ with area at most $(1+\eps)\opt$ such that $V(\sol)\subset G$ for some $G\in \cG$. Let $G$ be the square grid corresponding to the true circumscribed ellipse of $Q$. Consider the grid cells containing $V(\Popt)$; let $\sol$ be the convex hull of these cells. Notice that $|V(\sol)|\leq 4|V(\Popt)|$. Moreover, if $\delta$ is the cell diameter of $G$, then $\sol$ is covered by the Minkowski sum $\Popt \oplus B_\delta$, where $B_\delta$ is the disk of radius $\delta$. Thus we have that
\begin{align*}
    \Area(\sol)&\leq \Area(\Popt \oplus B_\delta)\\
    &= \Area(\Popt) + \peri(\Popt)\cdot\delta+ \Area(B_\delta)\\
    &\leq \Area(\Popt)+\delta\cdot 4\cdot 64/\eps + \delta^2\pi,
\end{align*}
where the last inequality uses that the perimeter of a convex polygon is less than the perimeter of any covering simple polygon: in this case the square $\sigma$ covers $\sol$.
Since $\Area(\Popt)\geq \Area(E')=\pi/4$ and $\delta=\sqrt{2}\cdot\frac{64/\eps}{\lfloor 2^{16}/\eps^3\rfloor}<\eps^2/2^9$, we can continue the inequality chain as follows:
\[\Area(\sol)\leq \Area(\Popt)+2^8\delta/\eps + \delta^2\pi <  \Area(\Popt) + \eps/2+ \eps/2^{16}<(1+\eps)\Area(\Popt),\]
concluding the proof.
\end{proof}
\section{An FPTAS for the minimum-perimeter problem of convex objects in the plane}
\label{sec:FPTAS}
%-------------------------------------------------------------------------------
Let $\cO$ be a set of $n$ convex objects in the plane for which we want to compute
a minimum-perimeter convex intersecting polygon. We assume 
that $\cO$ cannot be be stabbed by a single point---this is easy to test without increasing the total running time. %\as{This sentence seems incomplete.}
Since a minimum-perimeter intersecting polygon is necessarily convex, we will from now on drop the adjective ``convex'' 
from our terminology. We do this even when referring to convex intersecting polygons that
are not necessarily of minimum perimeter.
%\medskip

In the previous section we have seen that 
for any $\eps>0$,
we can find a feasible 
rectangle $R$ and a square $\sigma$
%set~$\cC$ of $O(n^2)$  squares 
with the following property: 
%there is a square~$\sigma^*\in\cC$ 
%such that there exists a minimum-perimeter overlapping set
%$\Popt \subset \sigma^*$ with $\diam(\Popt) \geq \frac{1}{22}\cdot\diam(\sigma^*)$,
Either $\peri(R) \le (1+\eps)\opt$, or $\Popt\subseteq \sigma$ with $\diam(\Popt) = \Omega(1)\diam(\sigma)$.
Next we describe an algorithm that, given a parameter $\eps>0$ and a corresponding square~$\sigma$, %\in\cC$,
outputs % the following:
an intersecting polygon $C^*\subseteq \sigma$ for $\cO$ such that if $\peri(R)\ge(1+\eps)\peri(\Popt)$ then $\peri(C^*)\leq (1+\eps)\opt$, where $\opt=\peri(\Popt)$ (cf. Lemma~\ref{lem:dum_locate} and Equation~\ref{eq:lb_diam}).
%\begin{itemize}
%\item If $\sigma=\sigma^*$, then the algorithm outputs an overlapping set $C^*$ for $\cO$
%      such that $\peri(C^*)\leq (1+\eps)\cdot \opt$, where $\opt=\peri(\Popt)$.
%\item If $\sigma\neq\sigma^*$, then the algorithm outputs an overlapping set $C^*\subset \sigma$ 
%      for $\cO$ if it exists, and reports that such a set does not exist otherwise.
%\end{itemize}
%After running this algorithm on all squares~$\sigma\in\cC$, we report the 
%smallest-perimeter overlapping set that we have found.
Finally, we output either $R$ or $C^*$, whichever has smaller perimeter.

Our algorithm starts by partitioning $\sigma$ into a regular grid~$\Gs$ 
of~$O(1/\eps^2)$ cells of edge length at most $(\eps/8)\cdot\opt$. 
We say that a convex polygon is a \emph{grid polygon} if its vertices
are grid points from~$\Gs$.
The following observation is standard, but for completeness we include a proof.
%-------------------------------------------------------------------------------
\begin{observation}
\label{obs:grid}
Suppose $\sigma$contains an optimal solution~$\Popt$.
Let $\Ps$ be a minimum-perimeter grid polygon that is an intersecting polygon for $\cO$.
Then $\peri(\Ps) \leq (1+\eps)\cdot \opt$.
\end{observation}
%-------------------------------------------------------------------------------
\begin{proof}
Let $U$ be the union of all grid cells that intersect~$\Popt$, and let $Z$ be a square 
of edge length $(\eps/4)\cdot\opt$, centered at the origin. Then $U$ is contained
in the Minkowski sum $\Popt \oplus Z$. Since $\Popt \oplus Z$ is convex, this
implies that the convex hull $\ch(U)$ of $U$, which is a grid polygon that
is an intersecting polygon for~$\cO$, is contained in $\Popt \oplus Z$. 
Because the perimeter of a Minkowski sum of
two convex objects equals the sum of the perimeters of the objects,
we have 
\[
\peri(\Ps) \leq \peri(\ch(U))\leq \peri(\Popt) + \peri(Z) = (1+\eps)\cdot \opt. \qedhere
\]
\end{proof}
%-------------------------------------------------------------------------------
Next we describe an algorithm to compute a minimum-perimeter grid polygon~$\Ps$
that is an intersecting polygon for $\cO$.
%\medskip

First, we ``guess'' the lexicographically smallest vertex $\vbot$ of~$\Ps$, see Figure~\ref{fig:dp-fptas-new}(i).
We can guess $\vbot$ in $O(1/\eps^2)$ different ways.
For each possible guess we will find the best solution (if it exists),
and then we report the best solution found over all guesses. 

Now consider a fixed guess for the lexicographically smallest vertex $\vbot$ of~$\Ps$. 
% We first check if $\cO$ admits an overlapping set that has $\vbot$ as a
% lexicographically smallest vertex; if so, it also has such an overlapping set 
% that is a grid polygon. 
% This is the case if and only if 
% $h^+\setminus \rho_0$ overlaps all objects in $\cO$, where $h^+$ is the closed
% half-plane above the horizontal line through $\vbot$ and $\rho_0$
% is the horizontal ray emanating from~$\vbot$ and pointing to the left.
% (The ray $\rho_0$ does not include its starting point~$\vbot$.) This can trivially
% be tested in $O(n)$ time, so from now on we assume that an overlapping set
% exists with $\vbot$ as lexicographically smallest vertex.
With a slight abuse of notation we will use $\Ps$ to denote a minimum-perimeter
grid polygon that is an intersecting polygon of $\cO$ and that has
$\vbot$ as lexicographically smallest vertex. (If the polygon $\Ps$
does not exist, the algorithm described below will detect this.)
We will compute $\Ps$ by dynamic programming. 
%\medskip

The vertices of~$\Ps$ are grid points in the region $h^+\setminus\rho_0$, 
where $h^+$ is the closed half-plane above the horizontal line through $\vbot$ 
and $\rho_0$ is the horizontal ray emanating from~$\vbot$ and pointing to the left.
Let $V$ be the set of such grid points, excluding $\vbot$. 
We first order the points from $V$ in angular order around~$\vbot$.
More precisely, for a point $v\in V$, let $\phi(v)$ denote 
the angle over which we have to rotate $\rho_0$ in clockwise direction
until we hit~$v$. For two points $v,w\in V$, we write
$v\prec w$ if $\phi(v)<\phi(w)$.
Let $V^+ := V\cup \{\vbot,\vbotc\}$, where $\vbotc$ is a copy of $\vbot$,
and define $\vbot\prec v$ and $v\prec \vbotc$ for all $v\in V$. The copy
$\vbotc$ serves to distinguish the start and the end vertex of the clockwise circular
sequence of vertices of~$\Ps$. Note that if $\vbot,v_1,\ldots,v_k,\vbotc$
denotes this circular sequence then $\vbot \prec v_1 \prec \cdots \prec v_k \prec \vbotc$ (since $C(\sigma)$ will never have two vertices that make the same angle with $\vbot$).
% \as{Reviewer1 is asking for more details. My suggestion: $v\prec w$ if $\phi(v) \leq \phi(w)$,
% namely $\leq$ instead of $<$, in order to incorporate collinear points. This would not give us a partial order though.
% Otherwise, we should assume that $C(\sigma)$ is minimal in the sense that there are no collinear points.
% (maybe the second choice is better/easier)} 

We now describe our dynamic-programming algorithm. 
Consider a polyline from $\vbot$ to some point $v\in V$. We say that this polyline
is a \emph{convex chain} if, together with the line segment $v \vbot$, it forms
a convex polygon. We denote the convex polygon induced by such a chain~$\Gamma$ by $\sol(\Gamma)$.
The problem we now wish to solve is as follows: 
%-----------------------------------------------------------------------------------
\begin{quote}
Compute a minimum-length convex chain $\Gamma^*$ from $\vbot$ to $\vbotc$
such that $\sol(\Gamma^*)$ is an intersecting polygon for $\cO$.
\end{quote}
%-----------------------------------------------------------------------------------
Our dynamic-programming algorithm uses the partial order~$\prec$ defined above.
We now want to define a subproblem for each point $v\in V^+$,
which is to find the ``best'' chain $\Gamma_v$ ending at~$v$.
For this to work, we need to know which objects from $\cO$ should be covered by the partial
solution~$\sol(\Gamma_v)$. This is difficult, however, because objects that intersect the ray
from $\vbot$ and going through~$v$ could either be intersected by~$\sol(\Gamma_v)$
or by the part of the solution that comes after~$v$. To overcome this problem
we let the subproblems be defined by the last edge on the chain, instead of
by the last vertex. This way we can decide which objects should be covered by a partial
solution, as explained next.
%-----------------------------------------------------------------------------------
\begin{figure}
\begin{center}
\includegraphics[width=\textwidth]{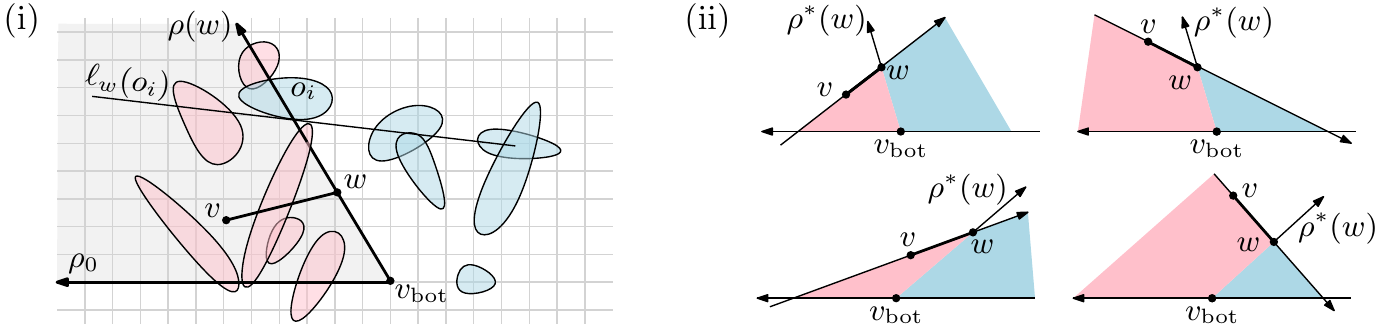}
\end{center}
\caption{(i) The wedge defined by $\rho_0$ and $\rho(w)$ is shown in light grey.
             Objects in $\cO(v,w)$ are red, objects in $\cO\setminus\cO(v,w)$ are blue.
        (ii) The partial solution $\sol(\Gamma)$ must be contained in the red region,
              while $C\setminus \sol(\Gamma)$ must lie in the blue region.
              There are four situations, depending on whether the angle between
              $\rho_0$ and $\rho(w)$ is acute or not, and whether the line containing~$vw$
              intersects $\rho_0$ or not.% \mdb{Not sure if part~(ii) of the figure adds much,
%              the cases are all the same.} \as{I think that I agree with your comment.}
}
\label{fig:dp-fptas-new}
\end{figure}
%-----------------------------------------------------------------------------------

Consider a convex chain from $\vbot$ to a point $w\in V^+$ whose last edge is~$vw$.
Let $\rho(w)$ be the ray emanating from $\vbot$ in the direction of~$w$,
and let $\rho^*(w)$ be the part of the ray starting at~$w$.
For $w=\vbotc$ we define $\rho(w)$ to be the horizontal ray
emanating from $\vbot$ and going to the right, and we define $\rho^*(w)=\rho(w)$.
For an object~$o_i\in\cO$ that intersects~$\rho^*(w)$, let $\ell_w(o_i)$ be a
% \as{This tangent is unique?} \mdb{changed ``the'' to ``a''}
line that is tangent to~$o_i$ at the first intersection point of $\rho^*(w)$ with~$o_i$.
We now define the set $\cO(v,w)$ to be the subset of objects $o_i\in\cO$
such that one of the following conditions is satisfied; see also Figure~\ref{fig:dp-fptas-new}(i).
%-----------------------------------------------------------------------------------
\begin{enumerate}[(i)]
\item $o_i$ intersects the wedge defined by $\rho_0$ and $\rho(w)$, but not $\rho(w)$ itself; or
      \label{cond:1}
\item $o_i$ intersects $\vbot w$; or
      \label{cond:2}
\item $o_i$ intersects $\rho^*(w)$ but not $\vbot w$, and the tangent line~$\ell_w(o_i)$ intersects
      the half-line containing~$vw$ and ending at~$w$.
      \label{cond:3}
\end{enumerate} 
%-----------------------------------------------------------------------------------
The next lemma shows that we can use the sets $\cO(v,w)$ to define our subproblems.
%-----------------------------------------------------------------------------------
\begin{lemma}
\label{lem:fptas-correctness}
Let $C$ be any convex polygon that is an intersecting polygon for~$\cO$ and that has
$\vbot$ as lexicographically smallest vertex and $vw$ as one of its edges.
Let $\Gamma_w$ be the part of~$\bd C$ from $\vbot$ to $w$ in clockwise direction. 
Then all objects in $\cO(v,w)$ intersect~$\sol(\Gamma_w)$ and all
objects in $\cO\setminus \cO(v,w)$ intersect~$C\setminus \sol(\Gamma_w)$.
\end{lemma}
%-----------------------------------------------------------------------------------
\begin{proof}
Because $C$ is convex, it must lie in $h^+$ (the closed half-plane above the horizontal
line though~$\vbot$) and in the region to the right 
of the line $\ell(v,w)$ through $v$ 
and $w$ and directed from $v$ to $w$. This region is split into two regions by~$\rho(w)$;
the region to the left of $\rho(w)$ contains $\sol(\Gamma_w)$ and the region
to the right of $\rho(w)$ contains $C\setminus \sol(\Gamma_w)$. 
See Figure~\ref{fig:dp-fptas-new}(ii) for various possible configurations. 

Consider an object~$o_i\in\cO(v,w)$. 
If $o_i$ is in $\cO(v,w)$ because of condition~(\ref{cond:1}), then it cannot intersect 
$C\setminus \sol(\Gamma_w)$ and, hence, it must intersect~$\sol(\Gamma_w)$.
If $o_i$ is in $\cO(v,w)$ because of condition~(\ref{cond:2}), then trivially
$o_i$ intersects $\sol(\Gamma_w)$. Finally, if $o_i$ is in $\cO(v,w)$ because of
condition~(\ref{cond:3}) then again it cannot intersect $C\setminus \sol(\Gamma_w)$, 
and so it must intersect~$\sol(\Gamma_w)$.

On the other hand, suppose that $o_i\not\in \cO(v,w)$. If $o_i$
does not intersect $\rho(w)$ then this means that $o_i$ does not
intersect the wedge defined by $\rho_0$ and $\rho(w)$. Hence, 
it cannot intersect $\sol(\Gamma_w)$ and so it must 
intersect~$C\setminus \sol(\Gamma_w)$. If $o_i$ intersects $\rho(w)$
then the reason it is not in $\cO(v,w)$ is that the tangent line~$\ell_w(o_i)$
does not intersect the half-line containing~$vw$ and ending at~$w$, 
which means the tangent separates $o_i$ from $\sol(\Gamma_w)$. 
Hence, $o_i$ must intersect $C\setminus \sol(\Gamma_w)$.
\end{proof}
%-----------------------------------------------------------------------------------
We can now state our dynamic program. To this end we define, for two points $v,w\in V^+$ 
with $v\prec w$, a table entry $A[v,w]$ as follows.
\\[2mm]
\begin{tabular}{ll}
$A[v,w]$ \ := 
& \begin{minipage}[t]{11.5cm}
   the minimum length of a convex chain~$\Gamma$ from $\vbot$ to $w$ 
   whose last edge is $vw$ and such that all objects in $\cO(v,w)$ intersect~$C(\Gamma)$,
   \end{minipage}
\end{tabular}\\[2mm]
where the minimum is~$\infty$ if no such chain exists.
Lemma~\ref{lem:fptas-correctness} implies the following.
%-----------------------------------------------------------------------------------
\begin{observation}\label{obs:global_opt_from_dp}
Let $\Gamma^*$ be a shortest convex chain from $\vbot$ to $\vbotc$
such that $\sol(\Gamma^*)$ is an intersecting polygon for $\cO$. Then 
$\mylength(\Gamma^*) = \min \{ A[v,\vbotc] : v\in V \mbox{ and } \cO(v,\vbotc) = \cO \}$.
\end{observation}
%-----------------------------------------------------------------------------------
\begin{proof}
Consider any $v\in V$ such that $\cO(v,\vbotc) = \cO$. By definition of~$A[v,\vbotc]$, 
there is a convex chain $\Gamma$ from $\vbot$ to $\vbotc$ such that $\sol(\Gamma)$ 
is an intersecting polygon for~$\cO$. 
Hence, $\mylength(\Gamma^*) \leq \min \{ A[v,\vbotc] : v\in V \mbox{ and } \cO(v,\vbotc) = \cO \}$.

Conversely, let $\Gamma^*$ be a minimum-length convex chain from $\vbot$ to $\vbotc$
such that $\sol(\Gamma^*)$ is an intersecting polygon for $\cO$. Let $v$ be the
vertex preceding $\vbotc$ on $\Gamma^*$. By Lemma~\ref{lem:fptas-correctness}
all objects in $\cO\setminus \cO(v,\vbotc)$ intersect $\sol(\Gamma^*) \setminus \sol(\Gamma^*_{\vbotc})$,
where $\Gamma^*_{\vbotc}$ is the part of $\Gamma^*$ from $\vbot$ to $\vbotc$.
But trivially  $\Gamma^* = \Gamma^*_{\vbotc}$ and so 
$\sol(\Gamma^*) \setminus \sol(\Gamma^*_{\vbotc})=\emptyset$. 
Since $\sol(\Gamma^*)$ is an intersecting polygon for $\cO$, this implies that
$\cO\setminus \cO(v,\vbotc)=\emptyset$. Hence, $\cO = \cO(v,\vbotc)$,
and so $\mylength(\Gamma^*) \geq \min \{ A[v,\vbotc] : v\in V \mbox{ and } \cO(v,\vbotc) = \cO \}$.
% \mdb{Please check proof; second part feels a bit convoluted.} 
% \mdb{The notation was a bit awkward, I tried to rephrase a bit. Please check.} \as{Proof looks correct to me.}
\end{proof}
%-----------------------------------------------------------------------------------
Hence, if we can compute all table entries $A[v,w]$ then we have indeed solved our problem. 
(The lemma only tells us something about the value of an optimal solution, but given the 
table entries $A[v,w]$ we can compute the solution itself in a standard way.)

The entries $A[v,w]$ can be computed using the following lemma. Define $\Delta(\vbot,v,w)$ 
to be the triangle with vertices $\vbot,v,w$.
%-----------------------------------------------------------------------------------
\begin{lemma} \label{lem:dp-fptas}
Let $v,w\in V^+$ with $v\prec w$.
Let $V(v,w)$ be the set of all points $u\in V\cup \{\vbot\}$ with $u\prec v$ 
such that $u$ lies below the line $\ell(v,w)$ through $v$ and $w$ and such that
all objects in $\cO(v,w) \setminus \cO(u,v)$ intersect $\Delta(\vbot,v,w)$. Then
\[
A[v,w] = \! \left\{ \begin{array}{ll}
                    |\vbot w| & \mbox{if $v\!=\!\vbot$ and all objects in $\cO(v,w)$ intersect $\vbot w$} \\
                    \infty & \mbox{if $v\!=\!\vbot$ and not all objects in $\cO(v,w)$ intersect $\vbot w$} \\
                    |vw|\!+\!\!\min\limits_{u\in V(v,w)} \!\! A[u,v] & \mbox{otherwise} 
                    \end{array}
             \right.
\]
\end{lemma}
%-----------------------------------------------------------------------------------
\begin{proof}
The first two cases immediately follow from the definition of~$A[v,w]$, since for
$v=\vbot$ we have $\Gamma = \vbot w$ and so $\sol(\Gamma)=\vbot w$. 

To prove the third case, let $\Gamma^*$ be a minimum-length convex chain from $\vbot$ to $w$ 
whose last edge is $vw$ and such that all objects in $\cO(v,w)$ intersect~$C(\Gamma^*)$. 
Let $A[v,w]$ be the value computed by the recursive formula given in the lemma.
We must prove that $A[v,w]=\mylength(\Gamma^*)$.

Let $u^*$ be the vertex preceding~$v$ on~$\Gamma^*$. 
Then $u^*\prec v$ and (because of convexity) $u^*$ must lie below $\ell(v,w)$. 
Now consider an object $o_i\in \cO(v,w) \setminus \cO(u^*,v)$. 
By Lemma~\ref{lem:fptas-correctness} it must be intersected by $\Delta(\vbot,v,w)$,
since $\Delta(\vbot,v,w) = \sol(\Gamma^*) \setminus \sol(\Gamma^*\setminus vw)$.
Hence, $u^*\in V(v,w)$. Moreover, $A[u^*,v]=\mylength(\Gamma^* \setminus vw)$ by induction.
Hence,
\[
A[v,w] := |vw|+\min\limits_{u\in V(v,w)} A[u,v] \leq  |vw|+ A[u^*,v] = \mylength(\Gamma^*).
\]

On the other hand, for any $u\prec v$ with $A[u,v]\neq\infty$
there is, by induction, a convex chain of length $A[u,v]$ starting at $\vbot$ 
whose last edge is $uv$ that intersects all objects in $\cO(u,v)$. 
If $u\in  V(v,w)$ then we can extend this chain to a convex chain~$\Gamma$ 
starting at $\vbot$ whose last edge is $vw$ and such that all objects 
in $\cO(v,w)$ intersect~$C(\Gamma)$. Hence, for any $u\in V(v,w)$
there is a convex chain $\Gamma$ from $\vbot$ to $w$ whose least edge is $vw$
and such that $\sol(\Gamma)$ intersects all objects in~$\cO(v,w)$.
Thus $A[v,w] := |vw|+\min\limits_{u\in V(v,w)} A[u,v]$ is at least
the minimum length of such a chain.

This finishes the correctness proof of the third case.
\end{proof}
%-----------------------------------------------------------------------------------
Putting everything together, we can finish the proof of Theorem~\ref{thm:2d-fptas-perimeter}.
%-----------------------------------------------------------------------------------
\begin{proof}[Proof of Theorem~\ref{thm:2d-fptas-perimeter}]
%The set $\cC$ of $O(n^2)$  candidate squares, one of which contains an optimal solution,
%can be computed in \mdb{$O(n^4)$} time by Lemma~\ref{lem:locateopt}. 
We first use algorithm $A1$ from~\cite{DumitrescuJ12} to compute the rectangle $R$ and the square $\sigma$, which as discussed can be done in $O(n^{2.373}/\eps)$ time.
For a square~$\sigma$ we guess the vertex $\vbot$ in $O(1/\eps^2)$ different ways. 

For each guess we run the dynamic-programming algorithm described above.
There are $O(1/\eps^4)$ entries $A[u,v]$ in the dynamic-programming table. The most time-consuming 
computation of a table entry is in the third case of Lemma~\ref{lem:dp-fptas}. Here we need to compute
the set $\cO(v,w)$, which can be done in $O(n)$ time by checking every $o_i\in \cO$.
For each of the $O(1/\eps^2)$ points with $u\prec v$ such that $u$ lies below the line~$\ell(v,w)$
we then check in $O(n)$ time if all objects in $\cO(v,w) \setminus \cO(u,v)$ 
intersect $\Delta(\vbot,v,w)$, so that we can compute $A[v,w]$. Hence, computing $A[v,w]$
takes $O(n/\eps^2)$ time, which implies that the whole dynamic program needs $O(n/\eps^6)$ time.

Thus the algorithm takes $O(n^{2.373}/\eps) + O(1/\eps^2) \cdot O(n/\eps^6) = O(n^{2.373}/\eps +n/\eps^{8})$ time.
\end{proof}

\begin{remark}
Although Theorem~\ref{thm:2d-fptas-perimeter} is stated only for the case where $\cO$ is a set of convex polygons, it is not too hard to extend it to other convex objects, for example disks: one just needs to replace the approximate rectangle-finding linear program of Dumitrescu and Jiang~\cite{DumitrescuJ12} with some other polynomial-time algorithm to find an (approximate) minimum perimeter intersecting rectangle in each of the $O(1/\eps)$ orientations.
\end{remark}

\section{An FPTAS for the minimum-area convex intersecting polygon}\label{sec:area-fptas}

Due to Lemma~\ref{lem:constantsize_orflat}, either there exists an approximate solution with at most $8$ vertices, or it is sufficient to compute the minimum-area convex intersecting polygon whose vertices are in a grid $G\in \cG$. Since we have no way to distinguish between these outcomes, we will compute a minimum feasible solution of at most $8$ vertices, as well as a minimum feasible polygon (if it exists) in each grid $G\in \cG$, and simply return the smallest area polygon that we have found.

In Section~\ref{sec:contantsizearea} we show how to compute an (approximate) minimum area polygon with at most $8$ vertices with known algebraic methods in $O(n^{17}\log(1/\eps))$ time; here we concentrate on adapting our dynamic programming for the minimum-perimeter problem. Let us fix a grid $G\in \cG$. Keeping the notations as before, we see that everything up to (and including) Lemma~\ref{lem:fptas-correctness} still holds. We define the subproblems as follows.\\[2mm]
\begin{tabular}{ll}
$A[v,w]$ \ := 
& \begin{minipage}[t]{11.5cm}
   the minimum area in the convex hull $C(\Gamma)$ of a convex chain~$\Gamma$ from $\vbot$ to $w$ 
   whose last edge is $vw$ and such that all objects in $\cO(v,w)$ intersect~$C(\Gamma)$,
   \end{minipage}
\end{tabular}\\[2mm]
where the minimum is $\infty$ if no such chain exists. The proof of the following observation is the same as the proof of Observation~\ref{obs:global_opt_from_dp}, one only needs to change all mentions of length or perimeter of a convex chain to the area of the convex hull of the chain.

\begin{observation}
Let $\Gamma^*$ be a convex chain from $\vbot$ to $\vbotc$ where $C(\Gamma^*)$ has minimum area and it is an intersecting polygon for $\cO$. Then $\Area(C(\Gamma^*))=\min \{A[v,\vbotc]:v\in V\text{ and }\cO(v,\vbotc)=\cO\}$.
\end{observation}

The recursion also works analogously: we only need to change the starting value, and in the recursive step we add the area of the triangle corresponding to the new segment instead of its length. The proof is again analogous to the minimum-area variant (to Lemma~\ref{lem:dp-fptas}).

\begin{lemma} \label{lem:dp-fptas-area}
Let $v,w\in V^+$ with $v\prec w$.
Let $V(v,w)$ be the set of all points $u\in V\cup \{\vbot\}$ with $u\prec v$ 
such that $u$ lies below the line $\ell(v,w)$ through $v$ and $w$ and such that
all objects in $\cO(v,w) \setminus \cO(u,v)$ intersect $\Delta(\vbot,v,w)$. Then
\begin{multline*}
    A[v,w] =
    \begin{cases}
                0 & \parbox{6cm}{if $v=\vbot$ and all objects\\ in $\cO(v,w)$ intersect $\vbot w$} \\[4mm]
                \infty & \parbox{6cm}{if $v=\vbot$ and not all objects\\ in $\cO(v,w)$ intersect $\vbot w$} \\[3mm]
                \Area(\Delta(\vbot,v,w)) + \min\limits_{u\in V(v,w)} A[u,v] & \mbox{otherwise.} 
   \end{cases}
\end{multline*}
\end{lemma}

Putting everything together, we can prove Theorem~\ref{thm:2d-fptas-area}.

\begin{proof}
First, we compute the approximately optimal intersecting polygon with at most $8$ vertices in $O(n^{17}\log(1/\eps))$ time by Theorem~\ref{thm:areaconstant}, see below.
The set $\cG$ of $O(n^{10})$ candidate grids can be computed in $O(n^{10})$ time. Each grid $G\in \cG$ has $O(1/\eps^6)$ vertices, among which we guess the vertex $\vbot$ in $O(1/\eps^6)$ time.

For each guess we run the dynamic programming, whose table has $O(1/\eps^{12})$ entries $A[u,v]$. Again
the most time-consuming part of computing a table entry is in the third case of
Lemma~\ref{lem:dp-fptas-area}: for each of the $O(1/\eps^6)$ points $u\in V(v,w)$ we check if all
objects in $\cO(v,w)\setminus \cO(u,v)$ intersect $\Delta(\vbot,v,w)$. Thus computing $A[v,w]$ takes
$O(n/\eps^6)$ time, and the dynamic program takes $O(n/\eps^{18})$ time. This part of the algorithm
therefore takes $O(n^{10})\cdot O(1/\eps^6)\cdot O(n/\eps^{18})=O(n^{11}/\eps^{24})$ time.
\end{proof}

\subsection{Approximating the minimum area polygon of constantly many vertices}\label{sec:contantsizearea}

Here we will show how that one can compute an approximate minimum-area intersecting polygon with $k\leq 8$ vertices; we will run this algorithm for each $k=2,3,\dots,8$. If $k=2$, then the optimum is a degenerate polygon (a segment) of area $0$. Since the length of the segment is irrelevant, it is in fact enough to check if there is a line that intersects all input objects. If there is such a line, then it can be rotated until it passes through two vertices of input objects without affecting feasibility. Thus, in $O(n^2)$ time, we can enumerate all pairs of vertices, and for each pair we can check in $O(n)$ time whether the line through them intersects all objects. Consequently, in $O(n^3)$ time we can decide if there is a feasible solution with $k=2$ (and $\opt=0$).

Assume now that $\opt>0$ and $3\leq k \leq 8$.
Notice that any intersecting polygon that has a vertex outside the union of the boundaries of $\cO$ cannot be optimal, as such a vertex can be moved to a boundary while maintaining feasibility and decreasing the area of the intersecting polygon. We begin our algorithm by guessing the $k$ input polygon sides on which these vertices lie. These sides can be written as $p_i+\lambda_i s_i$, where $p_i$ is an arbitrary point on the polygon side, $s_i$ is a unit vector parallel to the side, and $\lambda_i$ is a real number from some interval (where the interval may be unbounded on either side). We are then looking for a polygon with vertices of the form $v_i=p_i+\lambda_i s_i$ for each $i\in [k]$.

Next, we can express the feasibility of the solution using constant-degree polynomial inequalities over the variables $\lambda_i$. Indeed, for a fixed cyclic ordering of the vertices along the boundary of the optimum polygon, the convexity can be expressed by the fact that consecutive triplets have the same orientation. Applied for all cyclic orderings $\xi$, we get
\[\bigvee_{\substack{\xi:[k]\rightarrow [k]\\ \text{ cyclic order }}} \bigwedge_{j\in [k]} \quad \Big(\mathrm{det}(v_{\xi(j)}-v_{\xi(j+1)},\; v_{\xi(j)}-v_{\xi(j+2)})>0\Big),\]
where indices are modulo $k$ and $\mathrm{det}(a,b)$ is the determinant of the $2\times 2$ matrix with columns $a$ and $b$.

In order to check feasibility for each convex polygon $\cO$, we can write a Boolean expression of $O(n)$ constraints expressing that for each object $o\in \cO$ there is a side of $o$ and a side of the intersecting polygon that intersect. These constraints are again expressible as the signs of $2\times 2$ determinants involving the vertices of some $o\in \cO$ and of the intersecting polygon.

Finally, for a target real value $\mu>0$, we can express that the area of the intersecting polygon is at most $\mu$: the area of a polygon can be expressed as the sum of signed areas of the triangles formed by its sides and a designated vertex. The area is therefore a degree-2 polynomial of the variables $\lambda_i$.

Thus, we can express that there is a size-$k$ solution with a Boolean formula of $O(n)$ polynomial inequalities of degree $2$, where the variables $\lambda_1,\dots,\lambda_k$ are existentially quantified. We get a Boolean expression $\Phi_\mu$ of $O(n)$ polynomial inequalities of degree $2$ that has $k\leq 8$ existentially quantified variables (as we can reuse the variables for different values of $k$). Formula $\Phi$ is true if and only if there is a solution on $k$ vertices of area at most $\mu$.

Basu~\etal~\cite{BasuPR96} provide an algorithm to decide the truth of a formula with $t$ existentially quantified variables that has $s$ polynomial inequalities of degree at most $d$ with $s^{t+1}d^{O(t)}$ time, which is $O(n^9)$ time in case of $\Phi_\mu$. We can therefore use a binary search strategy to find a value $\mu$ such that $\Phi_{\mu/(1+\eps)}$ is false and $\Phi_{\mu}$ is true. One can observe that it is sufficient to search among the values $\mu=(1+\eps)^z$ where $z$ is an integer. In order to ensure that the strategy can begin and it terminates, it is necessary that we have an upper bound and a positive lower bound for $\opt$.

Recall that our input is a set of points with $c\log n$-bit coordinates for some constant $c$. We can imagine that these points are grid points from a grid of size $n^c\times n^c$, and we fix the unit to be the side length of a grid cell. Observe that $\opt\leq n^{2c}$ as the circumscribed square of the grid is a valid solution. We will now work towards a lower bound.

% The width of a set $C\subset \Reals^2$ in a direction $r$ is defined as $\max_{x\in C} \langle x,r \rangle-  \min_{x\in C} \langle x,r\rangle$, that is, the distance of the tangents of $C$ with normal $r$. The width of a set is simply defined as the minimum width over all directions $r$ (where $r$ is a unit vector). Note that if $C$ is convex, then $\Area(C)\geq \frac{\width(C)^2}{2}$.

\begin{lemma}\label{lem:arealowerbound}
Let $\Popt$ be a minimum area convex overlapping polygon for a set $\cO$ of convex polygons.
If $\Area(\Popt)>0$ then $\Area(\Popt)\geq\Omega(\frac{1}{n^{8c}})$.
\end{lemma}

\begin{proof}
If $\Area(\Popt)>0$, then $\Popt$ has at least $3$ edges and thus contains at least two vertices from $X$ by Lemma~\ref{lem:minvolumevertexlocation}. Let $u,v\in X\cap \bd \Popt$ be such points. Notice that $\bd \Popt$ also has a vertex $w$ on some edge $e$ of an input polygon that is not incident to $a$ or $b$. We claim that the triangle $uvw$ has sides of length at least $1/n^{O(n)}$. Notice that this is sufficient to show the lemma, as it implies that the triangle has area $1/n^{O(n)}$, and since the triangle is contained in $\Popt$, the same area lower bound holds also for $\Popt$.

Consider first the point $u$. It is either an input point (and thus it has integer coordinates) or it is the intersection of two polygon edges, each of which are defined as the line through two points form the integer $n^c\times n^c$ grid. Therefore, $u$ is the intersection of the lines $a_1x+b_1y=c_1$ and $a_2x+b_2y=c_1$, where $a_i,b_i,c_i\; (i=1,2)$ are integers between $0$ and $n^c$. Elementary geometry yields that the first coordinate of the intersection can be written in the form $p_u/q_u$, where $p_u$ and $q_u$ are quadratic expressions of $a_i,b_i,c_i\; (i=1,2)$ with coefficients $\pm 1$. We also write the first coordinate of $v$ as $p_v/q_v$. Thus if the first coordinate of $u$ and $v$ are different, then they differ by at least $1/q_uq_v\geq 1/O(n^{4c})$, so applying the same argument to the second coordinates, we get
$\dist(u,v)\geq \min(|u_x-v_x|,|u_y,v_y|) \geq \Omega(1/n^{4c})$.

Similarly, the minimum distance of $u$ to the line defined by the edge $e$ is an expression of the form $p_e/q_e$ where $q$ is a degree four polynomial with constant coefficients over the coordinates of $u$ and the coordinates of the points defining $e$. Thus we have that $\dist(u,w)\geq \dist(u,e)\geq \Omega(1/n^{4c})$. An analogous argument works for the distance of $v$ and $w$, hence the area of the triangle $uvw$ is at least $\Omega(1/n^{8c})$.
\end{proof}

\begin{theorem}\label{thm:areaconstant}
Let $\cO$ be a set of convex polygons of total complexity $n$ in $\Reals^2$ and let $\eps>0$. Suppose that the minimum area convex intersecting polygon of $\cO$ has area $\opt$. Then given $\eps$ and $\cO$, we can compute in $O(n^{17}\log(1/\eps))$ time a convex intersecting polygon $C$ of $\cO$ where $\Area(C)\leq (1+\eps)\opt$.
\end{theorem}

\begin{proof}
Recall that we guess the $k\leq 8$ input polygon sides, giving $O(n^k)$ options to start with. For a fixed choice of $k$ sides, let $\Phi_\mu$ be the formula described above with target value $\mu$. The truth of the formula for any $\mu>0$ can be decided in time $n^{O(1)}$. By our upper bound and by Lemma~\ref{lem:arealowerbound} the range of possible values for $\opt$ is between $\Omega(1/n^{8c})$ and $n^{2c}$, thus the number of distinct values to consider is $O(\log_{1+\eps}\frac{n^{2c}}{1/n^{8c}})=O(\frac{\log n}{\eps})$. Doing binary search over this many values leads to $O(\log\frac{\log n}{\eps})=O(\log\log n +\log(1/\eps))$ evaluations of the formula $\Phi$, which by Basu~\etal~\cite{BasuPR96} needs $O(n^9)$ time, resulting in an overall running time of $O(n^8)\cdot O(n^9\log(1/\eps))=O(n^{17}\log(1/\eps))$.
\end{proof}

% \section{Extending the minimum-perimeter algorithm to non-polygonal convex objects}\label{sec:extend-nonpoly}

% We use the same strategy as in Section~\ref{sec:contantsizearea} to compute a minimum intersecting rectangle $R$ for each of the $O(1/\eps)$ orientations considered by the original algorithm. Let us now fix an orientation, and a fixed target perimeter $\mu$. We can write a semi-algebraic formula $\Psi_\mu$ with four existentially quantified variables (the coordinates of the bottom left and top right corners) using $O(n)$ polynomial inequalities of constant degree. The formula $\Psi_\mu$ expresses feasibility and that the solution rectangle has perimeter at least $\mu$. (Note that the degree of the polynomials will depend on the complexity of the object boundaries, but those are assumed to be constant-degree algebraic curves.) By Basu~\etal~\cite{BasuPR96} such a formula can be evaluated in $n^{O(1)}$ time.

% In order to design a binary search strategy, we need to establish lower and upper bounds for the optimum. As in Section~\ref{sec:contantsizearea}, we use the $n^c\times n^c$ grid as a starting point. The square covering the grid has perimeter $4n^c$, which is always feasible, and thus an upper bound of $\opt$. For the lower bound, note that if not all input objects can be stabbed by one point, then there are two 

%-------------------------------------------------------------------------------
\section{An exact algorithm for the minimum-perimeter intersecting polygon of segments}\label{sec:exact}
%-------------------------------------------------------------------------------
We describe an exact algorithm to compute a minimum-perimeter intersecting object
for a set~$\cO$ of line segments\footnote{Although we describe our algorithm for non-degenerate segments, all our arguments also work if some (or all) of the segments are in fact lines, rays, or even points.} in the plane. Consider an optimal solution~$\Popt$,
and let $Y$ be the set of all  endpoints of the segments in~$\cO$.
%Recall that any vertex of $\Popt$ is either a point from $Y$ or a reflection \aaa{I removed the word "perfect". I think this does not need definition but let me know if you disagree}
% on a segment in~$\cO$. \as{This sentence looks weird.}
 The exact algorithm is based on 
two subroutines for the following two
problems. 
As before, whenever we talk about intersecting polygons, we implicitly 
require them to be convex. \new{Figure~\ref{fig:subroutines} illustrates the polygons computed by these subroutines.}

%-----------------------------------------------------------------------------------
\begin{figure}[t]
\begin{center}
\includegraphics{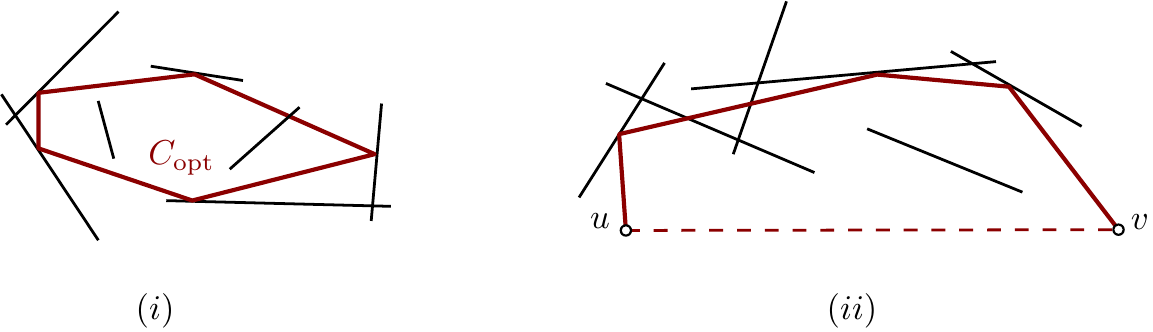}
\end{center}
\caption{\new{(i) Subroutine I computes a minimum perimeter intersecting polygon whose vertices are not segment endpoints.
             (ii) Subroutine II computes a polygon with fixed edge $uv$ for some (sub)set of segments. With the exception of $u$ and $v$, the polygon's vertices are not allowed to be segment endpoints.}}
\label{fig:subroutines}
\end{figure}
%-----------------------------------------------------------------------------------

\begin{itemize}
\item \emph{Subroutine~I}: 
      %Decide if $\cO$ admits an overlapping set all of whose vertices are
      %perfect reflections and, if so, compute a minimum-perimeter such overlapping set.
      %\mdb{Or is it: 
      %Compute a minimum-perimeter intersecting polygon for $\cO$  all of whose vertices are perfect reflections, or report that no such set exists.
      If $\cO$ admits a minimum-perimeter intersecting polygon with none of its vertices being in $Y$, then compute such a polygon. Otherwise compute a feasible intersecting polygon, or report $+\infty$.
      %}
   %   \mdb{Does it simply do the following: First compute an optimal solution for the set of lines containing the segments in~$S$. If
%    this solution also intersects all segments in~$S$ we are done. Otherwise, argue that there must be at 
%    least one segment endpoint that lies on the boundary of the optimal solution for $S$.}
\item \emph{Subroutine~II}:
      Given two points $u,v\in \Reals^2$ and a subset $\cO'\subseteq \cO$,
      decide if $\cO'$ admits a minimum-perimeter intersecting polygon that has $uv$ as one of its edges
      and none of whose other vertices belongs to $Y$ and, if so, compute a 
      minimum-perimeter such intersecting polygon. If no such minimum-perimeter intersecting polygon exists, report $+\infty$. Note that we allow $u=v$, in which case
      the edge $uv$ degenerates to a point.
\end{itemize}

\subsection{Subroutines~I and~II}
%\subsection{Structural Lemmas}
\label{sec:exact-subroutine}

 The goal of this subsection is to show the following theorem.
 %describe and analyze algorithms for solving the two subroutines. In particular, we show:
%The following theorem, also shown in Section~\ref{sec:exact-subroutine}, describes the running times of the subroutines.

\begin{theorem} \label{thm:subroutines}
There exist exact algorithms for \emph{Subroutine~I} and \emph{Subroutine~II} that run 
in time $O(n^6\log n)$ and $O(n^3\log n)$, respectively.
\end{theorem}

We want to establish that a portion of a minimum-perimeter polygon $\Gamma$ that does not contain any segment endpoints, is tractable in the sense that we can compute it with a black-box application of~\cite{tour_seq_polyg}. More specifically, we will prove Theorem~\ref{thm:subroutines} by invoking an algorithm for the Touring Polygons Problem~\cite{tour_seq_polyg} as a black box. For completeness we repeat the relevant definitions and results from~\cite{tour_seq_polyg}. We start by defining the Touring Polygons Problem (TPP). Note that in~\cite{tour_seq_polyg} a more general version is considered that requires the solution to remain within a fenced area, but we are only interested in the "unfenced" special case which slightly simplifies notation.

\begin{definition}[(unfenced) Touring Polygons Problem (TPP)\cite{tour_seq_polyg}]
 In the Touring Polygons Problem (TPP) we are given a sequence of convex polygons $P_1,P_2\dots P_k$, a starting point $s=P_0$, and an
 ending point $t=P_{k+1}$. We say that a path $\pi$ visits $P_i$, if $\pi\cap P_i\neq \varnothing$, and that $\pi$ visits the polygon sequence $P_1,\dots P_k$ if there exist points $q_i\in P_i$ such that the $q_i$'s appear in order of index along $\pi$. We seek to output the shortest path $\pi$ starting at $s$ and visiting the sequence of input polygons before finishing at $t$.
\end{definition}

Let $n$ be  the total complexity of the $k$ polygons input to TPP. Then:
\begin{theorem}[Theorem~$2$ in \cite{tour_seq_polyg}]
\label{thm:tpp} 
    The TPP for arbitrary convex polygons $P_i$ can be solved in time $O(nk^2\log n)$, using $O(kn)$ space.
\end{theorem}

Note that for segments and half-planes $k=O(n)$, and therefore TPP can be solved in $O(n^3\log n)$ time and quadratic space.

An important property of an optimal solution to TPP is \emph{uniqueness}:

\begin{lemma}[Lemma~$8$ in~\cite{tour_seq_polyg} (unfenced)]
\label{lem:unique-tour-seq-pol}
For any points $s,t\in\Reals^2$ there is a unique shortest path from $s$ to $t$ that visits the polygons $P_1,\dots P_k$ in order.
\end{lemma}

\paragraph*{Reduction}

We define a set of half-planes and an ordering on them, such that the optimum tour of the half-planes that respects the ordering gives a minimum intersecting polygon of the segments. 
In what follows, we will deal with a portion of the boundary of a minimum intersecting polygon (denoted by $\Gamma$), that is, a convex chain from some point $u$ to some point $v$. For the sake of simplicity, we rotate the polygon so that $u,v$ are on the $x$-axis, and the chain $\Gamma$ lies in the upper half-plane. Let $C(\Gamma)$ denote the polygon given by $\Gamma$ and the segment $uv$. The set of line segments to be intersected by $C(\Gamma)$ will be denoted by $\cO$.

Let us thus fix two points $u,v$ on the $x$-axis, and let $\cO$ be a set of $n$ segments. The \emph{half-plane of $o\in \cO$ with respect to $u,v$} is the (closed) half-plane defined
by $\Line(o)$ that (i) does not contain $u$ (see \Cref{fig:img1}), or (ii) if $u\in \Line(o)$, then it is the half-plane containing $v$,
or (iii) if $\Line(o)=\li{uv}$ then it is the upper half-plane defined by the $x$-axis. We sort all these half-planes of segments according to the direction of their
normal vectors in clockwise order, where the ordering of directions is set to start with $(0,-1)$. Let $\vec{n_i}$ be the normal vector of half-plane $H_i$, and let $\vec{n_0}=(0,-1)$ denote the normal of the half-plane $y\leq 0$. Let $o_i$ be the segment that was used to define~$H_i$. Finally, let $\mathcal{H}$ denote the set of half-planes $H_i$ for $i\geq 1$.
\begin{figure}[ht]
    \centering
    \includegraphics[scale=1]{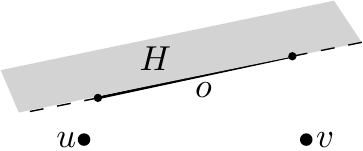}
    \caption{Illustration for case (i). }
    \label{fig:img1}
\end{figure}

%Consider an input instance $\cO$ for the subroutine described in the previous subsection. We describe how to construct a set of halfplanes $\mathcal{H}$ from $\cO$. Wlog we assume that the segments in $\cO = \{o_1,\dots o_n\}$ are indexed by the clockwise ordering of the orientation of their normals, starting from the orientation of $uv$. For any $o_i\in S$, let $H_i$ be the closed halfplane that does not contain point $u$ and is defined by the extension of segment $o_i$, and $\mathcal{H}=\cup_i H_i$. 

We first show that there exists a convex chain $\Gamma$ from $u$ to $v$ that visits all $H_i$ in order, i.e., there exist $t_i\in H_i$ such that $\Gamma$ visits the $t_i$'s in order. (In fact, we show the stronger statement that this holds for any convex chain from $u$ to $v$ that visits all half-planes.)

\begin{lemma}\label{lem:ordering}
    Fix the points $u, v$ on the $x$-axis and let $o_1,\dots,o_n$ be $n$ line segments, with corresponding half-planes $H_i$ where indices follow
    the ordering of the half-planes according to their normals as above. Consider a clockwise-oriented convex chain $\Gamma$ from $u$ to $v$ in the half-plane $y \geq 0$
    that visits each half-plane $H_i (i=1,\dots,n)$. Then there exist $n$ (not necessarily distinct) points $t_i \in H_i$, such that $\Gamma$ visits
    the points $t_i$ in the order of their indices.
\end{lemma}
\begin{proof}
    For a fixed path $\Gamma$ with the previous properties, the points $t_i$ will be defined with an iterative procedure. Let $j$ be a variable index initially set
    to one and establish a queue that contains all the half-planes in the specified order. Then the definition of the $t_i$ points comes from the following
    procedure:
    \begin{figure}[ht]
        \centering
        \includegraphics[scale=1]{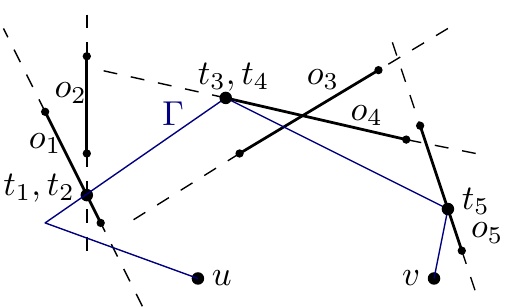}
        \caption{Illustration of the procedure.}
        \label{fig:img5}
    \end{figure}
    \begin{enumerate}
        \item Start from the point $u$ and follow the path $\Gamma$ until the first exit from a half-plane of the queue; suppose that this happens at a point $u'$. Let $k$ be the largest index of a half-plane that is being exited at $u'$.
        \item Set $t_j, \dots, t_k = u'$.
        \item Set $j$ equal to $k + 1$ and remove the half-planes with index less or equal to $k$ from the queue. As 
        long as the queue is not empty, repeat the first step with $u=u'$.
    \end{enumerate}
  %  \mdb{Working with exit points seems a bit counter-intuitive to me, since the starting point $u$
  %   lies outside all halfplanes. It seems more natural to me to define it as follows. For two points
  %   $p,q\in \Gamma$ define $p\prec q$ if $p$ is visited before $q$ in a clockwise traversal of $\Gamma$ starting at~$u$.
  %   For $i\geq 1$ define $t_i := t_{i-1}$ if $t_{i-1} \in H_i$ (where we set $t_0:=u$) and define $t_{i}$ to be the point where we
  %  enter $H_i$ otherwise. We now need to show that $t_{i-1}\preceq t_{i}$.}
    If the point $v$ is part of a half-plane $H_k$, then the path $\Gamma$ would never exit $H_k$. Therefore, there will be an iteration of the
    previous procedure, such that the path $\Gamma$ in the first step will meet the point $v$, while the queue will not be empty. In this case,
    the procedure stops and every remaining $t_i$ is set equal to $v$.
    
    Assuming that the $t_i$ points exist, by the construction, the path $\Gamma$ visits them in the specified order. Therefore it is sufficient to show
    that the points $t_i$ exist and that $t_i\in H_i$ for all $i$. Consider an iteration and let $j$, $k$, $u$ and $u'$ be the corresponding
    variables of the first two steps of the iteration.
    
    First, we show that for every $i: j \leq i \leq k$, it holds that $t_i \in H_i$. Suppose to the contrary that $t_i =u' \notin H_i$. Since $\Gamma$ is convex % \as{and remains in the halfplane $y \geq 0$ ?} 
    and has clockwise orientation, the rest of $\Gamma$
    remains in the interior of the cone defined by the segment $u'v$ and $\bd H_k$ (see \Cref{fig:img2}),
    % \mdb{Perhaps add argument: Indeed, if $\Gamma$ would enter $H_k$ again at a point between
    % $u'$ and $v$, then the part of $\Gamma$ from $u$ to $u'$ would be inside $H_i$,
    % contradicting the definition of $H_i$. [Well, at least in the case where $u\not\in H_i$;
    % the special cases where $\in \bd H_i$ require attention.]}\aaa{Mark, I don't follow your argument here but perhaps I am missing something? The idea is that no convex chain connecting $u$ with $v$, that exits $H_k$ at $u'$ can have a point between $u'$ and $v$ that is outside that cone, as that would contradict its convexity. Writing that such a chain would either re-enter $H_k$ which breaks convexity, or be to the right of segment $u'v$ which also must break convexity, seems like a repetition of what we already write. Would you like me to elaborate more on each of these two cases?}
    which is disjoint from $H_i$ (since (i) $i\le k$, thus $\angle (n_0,n_i) \le \angle (n_0,n_k)$, and (ii) $u'\not\in H_i$). 
 %   \mdb{Why is it disjoint? It seems to me that this uses the ordering on the
  %  halfplanes (and thus needs more explanation).}\aaa{I added the bit in the parenthesis}
    Therefore $\Gamma$ cannot intersect $H_i$  after
    leaving $u'$. Hence from the feasibility of $\Gamma$, the path must exit $H_i$ before it exits $H_k$, but this contradicts the definition of $u'$. 
    \begin{figure}[ht]
        \centering
        \includegraphics[scale=1]{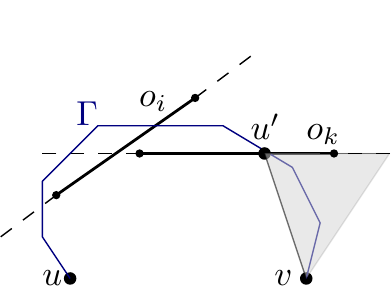}
        \caption{The first exit from a half-plane is not at point $u'$. This is a contradiction.}
        \label{fig:img2}
    \end{figure}

    Next, we show that every half-plane of index larger than $k$ will be visited by $\Gamma$ after the point $u'$. Suppose
    to the contrary that there exists a half-plane $H_i$ with $i > k$ which is not visited by $\Gamma$ after the point $u'$. In particular, we have that $v \not\in H_i$. From the feasibility
    of $\Gamma$, we can conclude that the path $\Gamma$ has exited from $H_i$ before or exactly at $u'$, but this contradicts
    the definition of $u'$ and  $k$. 
    This concludes the proof.%
\end{proof}

If we run the TPP-algorithm from~\cite{tour_seq_polyg} on $u,\cH,v$ (where $\mathcal{H}$ is ordered by the order specified above), then we obtain as output the shortest path $P$ from $u$ to $v$ that visits all half-planes in the given order. It is still possible though for the polygon defined by $P$ and $uv$ to not overlap $o_i$ (although $P$ visits $H_i$). However we claim that if the optimal solution $\Gamma$ 
does not go through any segment-endpoint then $P$ is not only feasible but also optimal for \emph{Subroutine~I} on this input. To show this, we need the following lemma, which shows that the convex combination of two different-length paths with the same ordering must be strictly shorter than the longest of the two paths.

\begin{lemma}\label{lem:convexcombo}
    Let $\Gamma_0 = \{p=a_0, a_1, \dots, a_n=q\}, \; \Gamma_1 = \{p=b_0, b_1, \dots, b_n=q\}$ be two paths such that $|\Gamma_0| < |\Gamma_1|$. 
    For a fixed $\lambda$, construct a path $\Gamma_\lambda = \{p=c_0, c_1, \dots, c_n=q\}$, where $c_i = \lambda \cdot b_i + (1 - \lambda) \cdot a_i$. Then 
    for every $\lambda: 0 \leq \lambda < 1$, it holds that $|\Gamma_\lambda| < |\Gamma_1|$.
\end{lemma}
\begin{proof}
    % The lengths of the three paths are equal to:
    % \begin{align*}
    %     |\Gamma_1| = |a_1 - p| + |q - a_n| + \sum_{i=1}^{n-1} |a_{i+1} - a_i| \\
    %     |\Gamma_2| = |b_1 - p| + |q - b_n| + \sum_{i=1}^{n-1} |b_{i+1} - b_i| \\
    %     |\Gamma_3| = |c_1 - p| + |q - c_n| + \sum_{i=1}^{n-1} |c_{i+1} - c_i| 
    % \end{align*}
%
    % First, we bound the first two terms of $|\Gamma_3|$:
    % \begin{align*}
    %     &|c_1 - p| + |q - c_n|\\
    %     &=|\lambda  a_1 + (1 - \lambda)  b_1 - p| + |q - (\lambda  a_n + (1 - \lambda)  b_n)| \\
    %     &=|\lambda  (a_1 - p) + (1 - \lambda)  (b_1 - p)|\\
    %     & \qquad + |\lambda  (q - a_n) + (1 - \lambda)  (q - b_n)| \\
    %     &\leq \lambda  |a_1 - p| + (1 - \lambda)  |b_1 - p|+ \lambda  |q - a_n| + (1 - \lambda)  |q - b_n|\\
    %     &=\lambda  (|a_1 - p| + |q - a_n|) + (1 - \lambda)  (|b_1 - p| + |q - b_n|).
    % \end{align*}
%    
    We can bound $|\Gamma_\lambda|$ in the following way:
    \begin{align*}
        |\Gamma_\lambda|&=\sum_{i=0}^{n-1} |c_{i+1} - c_i| \\
        &=\sum_{i=0}^{n-1} |\lambda  b_{i+1} + (1 - \lambda)  a_{i+1} - \lambda  b_i - (1 - \lambda)  a_i| \\
        &=\sum_{i=0}^{n-1} |\lambda  (b_{i+1} - b_i) + (1 - \lambda)  (a_{i+1} - a_i)| \\
        &\leq \sum_{i=0}^{n-1} \Big(|\lambda  (b_{i+1} - b_i)| + |(1 - \lambda)  (a_{i+1} - a_i)|\Big) \\
        &=\lambda  \sum_{i=0}^{n-1} |b_{i+1} - b_i| + (1 - \lambda)   \sum_{i=1}^{n-1} |a_{i+1} - a_i|.\\
        &=\lambda  |\Gamma_1| + (1 - \lambda)  |\Gamma_0|\\
        &< |\Gamma_1|.\qedhere
    \end{align*}
\end{proof}

The next lemma allows us to invoke the algorithm of~\cite{tour_seq_polyg} to compute a polygon portion, assuming that the portion does not contain segment endpoints.

\begin{lemma} \label{lem:endpointless}
    Let $C(\Gamma)$ be a minimum intersecting polygon for a set $\cO$ of segments. Let $\cH$ be the set of half-planes corresponding to $\cO$, and let $\Gamma_\cH$ be the shortest path from $u$ to $v$ that visits all  half-planes in $\cH$ in the order defined by their indices, i.e., there are points $t_i\in H_i$
    %\as{$t_i$ instead of $p_i$, as before?} 
    so that $\Gamma_\cH$ goes through the points $t_i$ in the order of their indices. Suppose that $\Gamma$ does not contain any of the endpoints
    of the segments (with the possible exception of $u$ and $v$). Then $\Gamma = \Gamma_\cH$.
\end{lemma}

\begin{proof}
    By Lemma~\ref{lem:ordering} and the convexity of $\Gamma$, there are $n$ points $a_i \in H_i$ for $\Gamma$ such that the path $\Gamma$ visits them in the specified order. We claim that there are no other vertices on $\Gamma$, that is, it can be represented as the convex chain
    \[
        \Gamma = (u, a_1, \dots, a_n, v), \text{where } a_i \in H_i. \\
    \]
    Indeed, if $v$ %\as{I suppose $v$ is not the same as the last vertex of $\Gamma$}
    is a vertex of $\Gamma$ that is not a visit point, then $v$ cannot be the only point where some $o_i$ is visited, so we can shorten $\Gamma$ by shortcutting from the last point of $\bigcup_i o_i$ before $v$ to the first point of $\bigcup_i o_i$ after $v$ (in case $\Gamma$ and some $o_i$ have a shared segment $s$ starting/ending at $v$, then we use the starting/ending point of $s$ as the starting/ending point of the shortcut).
    
    Since $\Gamma$ does not contain endpoints apart from $u$ and $v$, all of its internal vertices are in the interior of some segment in $\cO$ that is tangent to $C(\Gamma)$ by a similar shortcutting argument. Notice that if $\Gamma$ visits some segment $o_i$, that is tangent to $C(\Gamma)$ then $H_i$ intersects $\Gamma$ only at this point; consequently, this is the only location where the point $a_i$ could be located. To summarize, at each vertex of $\Gamma$ there must be some segment $o_i$ whose half-plane $H_i$ can be visited only there, and therefore all vertices of $\Gamma$ are contained in the sequence $(u, a_1, \dots, a_n, v)$.
    
    By the definition of $\Gamma_\cH$, it contains points $u,b_1,\dots,b_n,v\; (b_i \in H_i)$ in this order; we claim that it can be represented as the chain
    \[\Gamma_\cH = (u, b_1, \dots, b_n, v), \text{where } b_i \in H_i.\] 
    This is easy to see: the chain $(u, b_1, \dots, b_n, v)$ visits all the half-planes in the correct order and any tour containing these points in this sequence cannot be shorter than this chain. 
    
    Continuing the proof, suppose that $\Gamma \neq \Gamma_\cH$. The path $\Gamma_\cH$ is the shortest path visiting the polygons $u,H_1,\dots,H_n,v$ in this order, so by Lemma~\ref{lem:unique-tour-seq-pol}, 
    $\Gamma_\cH$ is the unique such shortest path. The path $\Gamma$ is a different chain visiting each of these polygons in the same order,
    therefore $|\Gamma_\cH| < |\Gamma|$. Since $\Gamma$ does not visit any endpoint 
    of the segments in $\cO$, any $u\rightarrow v$ curve where each point $a_i$ is moved by a small amount towards some other point in $H_i$ is also feasible. In particular, it follows that there exists a small $0<\lambda<1$ such that the path
    \[\Gamma' = \{u, c_1, \dots, c_n, v\}\text{, where }c_i = \lambda  a_i + (1 - \lambda) b_i\]
    gives a feasible polygon $C(\Gamma')$ for $\cO$.  By \Cref{lem:convexcombo}, it holds that $|\Gamma'| < |\Gamma|$,  
    which contradicts the optimality of $\Gamma$.
\end{proof}

We are now ready to prove Theorem~\ref{thm:subroutines}.
\begin{proof}[Proof of Theorem~\ref{thm:subroutines}]
First we provide an algorithm for \emph{Subroutine~I}. If the optimum intersecting polygon's boundary does not go through any segment endpoints, then we use Lemma~\ref{lem:endpointless} for an arbitrary vertex $u$ of the optimum as a designated start- and endpoint. It follows that the boundary of the minimum perimeter intersecting polygon for $\cO$ and the optimum tour of the half-planes corresponding to $\cO$ in the order defined by their normals is identical.

To find this point $u$ we use the algorithm from~\cite{DBLP:journals/dcg/CarlssonJN99}. The algorithm~\cite{DBLP:journals/dcg/CarlssonJN99} is designed for the so-called \emph{watchman tour} problem, but this is equivalent to our setting, see~\cite{Jonsson02} and the relevant discussion in~\cite{tour_seq_polyg}. The algorithm of~\cite{DBLP:journals/dcg/CarlssonJN99} computes $O(n^3)$ candidate points (called \emph{event points} in the terminology of \cite{DBLP:journals/dcg/CarlssonJN99}) in $O(n^3)$ time, 
while providing the guarantee that the optimal tour will go through at least one of them. 
By Theorem~\ref{thm:tpp} for each of these $O(n^3)$ points a TPP solution can be computed in $O(n^3\log n)$ time. Then for each TPP solution we can check if the corresponding polygon is convex and intersects all the segments in $O(n^2)$ time. In the end, we return the minimum-perimeter feasible polygon found, or $+\infty$ if none of the $O(n^3)$ polygons are feasible. If the optimum polygon has no segment endpoints as vertices, then the returned polygon is the optimum. The algorithm takes $O(n^3)\cdot (O(n^3\log n)+O(n^2))=O(n^6\log n)$, time, which concludes the proof for \emph{Subroutine~I}.

With respect to \emph{Subroutine~II} first note that if $u\neq v$, then all that is required is a single call to the TPP algorithm with a running time of $O(n^3\log n)$. Furthermore, if $u=v$, then Lemmas~\ref{lem:ordering}, \ref{lem:convexcombo} and~\ref{lem:endpointless} all still work. Consequently, if the optimum tour of the complete set $\cO$ of segments contains only one segment endpoint $u$, then the minimum intersecting polygon for the polygons $u,H_1,\dots H_n,u$ (which is a single call to the TPP algorithm as well) is also optimal for $\cO$. Here Lemmas~\ref{lem:unique-tour-seq-pol}, \ref{lem:ordering}, and~\ref{lem:endpointless} give a way to detect whether the optimal polygon for a given set of segments and points $u,v$ is equal to the TPP solution. Indeed, if the polygon corresponding to the TPP solution is feasible (the corresponding polygon intersects all segments in $\cO$) and its vertices apart from $u,v$ are not from $Y$, then it must be the optimum. Consequently, if the polygon given by the TPP solution is non-convex, or if it fails to intersect some polygon, or if it contains a segment endpoint, then we can return $+\infty$ as the cost of the respective tour. The feasibility checks can be done in $O(n^2)$ time, thus the dominant part of the running time is $O(n^3\log n)$.
\end{proof}

%-----------------------------------------------------------------------------------
%\mypara{Setting the stage for the dynamic program.}
\paragraph*{Setting the stage for the dynamic program.}
%-----------------------------------------------------------------------------------
With Subroutine~I available, it remains to find the minimum-perimeter 
intersecting polygon at least one of whose vertices is a point from~$Y$. 
To this end we will develop an algorithm that, for a given point $\pbot\in Y$, 
finds a minimum-perimeter intersecting polygon that has $\pbot$ as a vertex (if it exists). 
We will run this algorithm for each choice of $\pbot\in Y$. Note that when
we do so, we may ignore all segments that contain~$\pbot$, as they
will intersect any intersecting polygon with $\pbot$ as a vertex.
\medskip

Let $\pbot\in Y$ be given. We first check if $\cO$ admits an intersecting polygon
that has $\pbot$ as a vertex. This is the case if and only if there exists
a half-plane with $\pbot$ on its boundary that overlaps all segments in $\cO$.
The existence of such a half-plane can easily be tested in $O(n\log n)$
using a rotational sweep around $\pbot$. So from now on we assume that 
an intersecting polygon exists that has $\pbot$ as a vertex. 
Let $\Popt$ be a minimum-perimeter such intersecting polygon. 
The boundary of $\Popt$ consists of chains that connect points from $Y$
and such that the interior vertices of these chains are disjoint from $Y$.
(If $\pbot$ is the only point from $Y$ that is a vertex of $\Popt$, then there is
a single chain connecting $\pbot$ to itself.) We will develop a dynamic-programming
algorithm similar to the algorithm from Section~\ref{sec:FPTAS}. The dynamic program
will be based on the points in $Y$ (instead of on grid points) and we will
use Subroutine~II to find the chains connecting the points from $Y$ on~$\partial \Popt$.
As we will see, however, there are several challenges that we
need to overcome to adapt the dynamic program.

Observe that if we take three rays emanating from $\pbot$ that are at a 120-degree
angle from each other, then at least one of them lies fully outside $\Popt$ (except for
the starting point~$\pbot$). Consider such a ray~$\rho_0$, and assume without loss of 
generality that $\rho_0$ is a horizontal ray pointing leftwards. Our goal is now to
find a minimum-perimeter intersecting polygon with $\pbot$ as a vertex and that is
not intersected by~$\rho_0$. As mentioned, we will do this with a dynamic program
similar to the one from Section~\ref{sec:FPTAS}.
The point~$\pbot$ will play the role of $\vbot$---but note that $\pbot$ need not be the lexicographically
smallest vertex of $\Popt$---and $Y\setminus \{\pbot\}$ will play the role of~$V$. % \as{I think we call it $V$ in the FPTAS}. 
Hence,
it is convenient to redefine $Y$ as $Y\setminus \{\pbot\}$ and $Y^+ := Y \cup \{\pbot,\pbotc\}$,
where $\pbotc$ is a copy of~$\pbot$. We define a (partial) angular order~$\prec$ on $Y^+$, as before.

In Section~\ref{sec:FPTAS}, we knew that $\vbot$ was the lexicographically smallest vertex
of~$\Ps$, and so we were looking for solution that lies above the horizontal line through $\vbot$. This was important to be able to
decide which objects should be intersected by a partial solution; see Figure~\ref{fig:dp-fptas-new}. In the current setting,
however, the point~$\pbot$ need not be the lexicographically smallest point. 
Hence, we also need to guess the orientation of the line tangent to $\Popt$ at $\pbot$.
To this end, let $\psi(o_i)$ be the angle over which we have to
rotate~$\rho_0$ in clockwise direction to make it parallel 
%\as{what if the segment is a point? We assume that it's always parallel to $\rho_0$?} 
%\mdb{Good point. We could either say that for simplicity we assume that
%the segments are not degenerate. But I think that what you suggest, which is to define
%$\psi(o_i)=0$ when $o_i$ is a point, should work as well.}
to a given segment~$o_i\in \cO$, and sort the set of (distinct) angles $\psi(o_i)$ in increasing order.
Let $\Psi := \{ \psi_1,\psi_2,\ldots \}$, where $\psi_j < \psi_{j+1}$
for all $1\leq j< |\Psi|$, denote this sorted set. Define $\psi_0 := 0$ and $\psi_{|\Psi|+1} = \pi$.
The problem we now wish to solve is
\begin{quotation}
\noindent Given a point $\pbot$ and value $j$ with $0\leq j\leq |\Psi|$, find a minimum-perimeter
          intersecting polygon $\Popt$ for $\cO$ such that
          \begin{itemize}
           \item $\pbot$ is a vertex of $\Popt$,
           \item the horizontal ray~$\rho_0$ going from $\pbot$ to the left 
                 does not intersect $\Popt$,
          \item $\Popt$ has a tangent line $\ell$ at $\pbot$ such that the clockwise angle from
                 $\rho_0$ to $\ell$ lies in the range $[\psi_{j},\psi_{j+1}]$.
          \end{itemize}
\end{quotation}

\paragraph*{The dynamic program.}
%-----------------------------------------------------------------------------------
We will now develop our dynamic program for the problem that we just stated, for a given
point $\pbot \in Y$, a ray $\rho_0$, and range $[\psi_{j},\psi_{j+1}]$; see Figure~\ref{fig:orientation-new}(i).
%for an illustration. 
%-----------------------------------------------------------------------------------
\begin{figure}
\begin{center}
\includegraphics{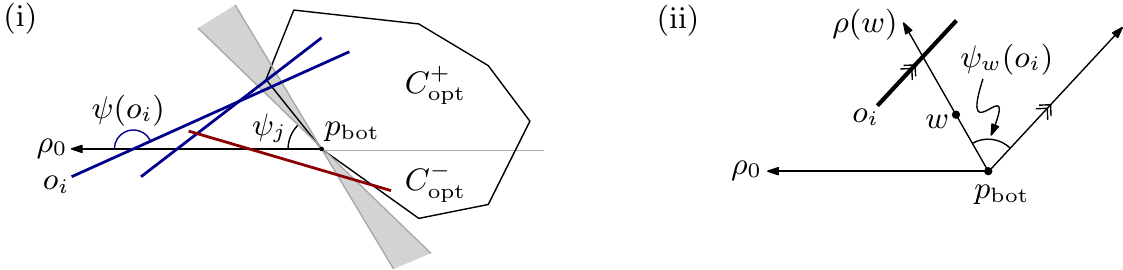}
\end{center}
\caption{(i) The grey double wedge indicates the region containing the tangent
at $\pbot$. (Note that then there must be more segments than the blue and red segments that are shown. In particular, there must be segments parallel to the lines delimiting the double wedge.)
The blue segments are in $\cO(\rho_0)^+$ so they must be intersected by $\Popt^+$, while the
             red segment is in $\cO(\rho_0)^-$ and must be intersected by $\Popt^-$.
             (ii) The definition of $\psi_w(o_i)$.}
\label{fig:orientation-new}
\end{figure}
%-----------------------------------------------------------------------------------

Let $\cO(\rho_0)\subseteq \cO$ denote the set of segments that intersect the ray~$\rho_0$, and 
let $\ell_0$ be the line containing~$\rho_0$. The line~$\ell_0$ may split the optimal solution~$\Popt$ 
into two parts: a part $\Popt^+$ above~$\ell_0$ and a part $\Popt^-$ below~$\ell_0$. 
Let $\psi(o_i)$ denote the angle over which we have to rotate 
$\rho_0$ in clockwise direction until it becomes parallel to $o_i$.
Since we have fixed the range of the tangent at $\pbot$ to lie in the range~$[\psi_j,\psi_{j+1}]$, 
we can split $\cO(\rho_0)$ into two subsets, %as follows: 
$\cO(\rho_0)^+ := \{ o_i \in \cO(\rho_0) :\psi(o_i) \geq \psi_{j+1} \}$ and 
$\cO(\rho_0)^- := \{ o_i \in \cO(\rho_0) :\psi(o_i) \leq \psi_j \}$.

%\begin{itemize}
%\item $\cO(\rho_0)^+ := \{ o_i \in \cO(\rho_0) :\psi(o_i) \geq \psi_{j+1} \}$
%\item $\cO(\rho_0)^- := \{ o_i \in \cO(\rho_0) :\psi(o_i) \leq \psi_j \}$.
%\end{itemize}
Note that $\cO(\rho_0)= \cO(\rho_0)^+\cup \cO(\rho_0)^-$,
because $\psi_j$ and $\psi_{j+1}$ are consecutive angles in~$\Psi$. 
% THE NEXT SENTENCE SEEMS WRONG.
% Moreover, for $o_i\in \cO(\rho_0)$ we cannot have $\psi(o_i)=\psi_j$
% or $\psi(o_i)=\psi_{j+1}$, because then $o_i$ cannot intersect $\Popt$. 
% \as{Is that true? In Fig2(i) for the red segment $o_\text{red} \in \cO(\rho_0)$, holds that $\psi(o_\text{red}) = \psi_{j - 1}$}
%\mdb{Changed $>$ and $<$ in the def of $\cO(\rho_0)^+$ and $\cO(\rho_0)^-$
%to $\geq$ and $\leq$, because, as Antonis noted,  segments
%in $\cO(\rho_0)$ may have an angle equal to $\psi_{j}$ or $\psi_{j+1}$. But the next
%sentence should still be true. There was also some inconsistency in
%using $[\psi_{j-1},\psi_j$] vs $[\psi_j,\psi_{j+1}]$. I now tried to use the latter one. Please check. } \as{Looks good to me.}
Because the orientation of the tangent at $\pbot$ lies in the range $[\psi_j,\psi_{j+1}]$,
we know that the segments in $\cO(\rho_0)^+$ must be intersected by $\Popt^+$, while the
segments in $\cO(\rho_0)^-$ must be intersected by $\Popt^-$; see Figure~\ref{fig:orientation-new}(i).
Intuitively, the segments in $\cO(\rho_0)^+$ must be intersected by ``the initial part'' of $\Popt$, while
the segments in $\cO(\rho_0)^-$ are intersected by ``the later part''. We will use this
when we define the subproblems in our dynamic program.

In Section~\ref{sec:FPTAS} we defined subproblems
for pairs of grid points~$v,w$. The goal of such a subproblem was to find the
minimum-length convex chain $\Gamma$ such that $\sol(\Gamma)$ intersects a certain subset
$\cO(v,w)$ and whose last edge is~$vw$. The fact that we knew the last edge~$vw$
was crucial to define the set $\cO(v,w)$, since the slope of $vw$ determined
which objects should be intersected by $\sol(\Gamma)$. 
In the current setting this does not work: we could define a subproblem for pairs
$v,w\in Y$, but ``consecutive'' vertices $v,w$ from $Y$ along $\Gamma$ are now
connected by a polyline $Z_{vw}$ whose inner vertices are disjoint from $Y$. %perfect reflections. 
The difficulty is that the polyline $Z_{vw}$ depends on the segments that need to be 
intersected by $\sol(\Gamma)$. Hence, there is a cyclic dependency between the set of
segments to be intersected by $\sol(\Gamma)$ (which depends on the slope of $zw$,
where $z$ is the vertex of $Z_{vw}$ preceding~$w$) and the vertex~$z$ 
(which depends on the segments that need to be 
intersected by $\sol(\Gamma)$). We overcome this problem as follows.

Similarly to the previous section, we call  a polyline $\Gamma$ from $\pbot$ to some point $v\in Y$ 
a \emph{convex chain} if, together with the line segment $v\pbot$, it forms a convex polygon.
We denote this polygon by $\sol(\Gamma)$.
Consider a convex chain $\Gamma_w$ ending in a point $w\in Y$. 
% Let $z$ be the vertex preceding $w$ (which is unknown and need not be a point in~$X$).
Let $\rho(w)$ denote the ray from $\pbot$ through $w$ and let $\rho^*(w)$ be the part of this ray 
starting at~$w$. Let $\cO(w)$ be the set of input segments that intersect~$\rho^*(w)$.
Of those segments, $\sol(\Gamma_w)$ must intersect the ones
such that the line~$\ell(o_i)$ containing\footnote{Since the input objects are now segments, the tangent
line~$\ell_w(o_i)$ is just the line containing~$o_i$.} the segment~$o_i$
intersects the half-line containing~$zw$ and ending at~$w$, where $z$ is the (unknown) vertex preceding~$w$.
For a segment $o_i\in\cO(w)$, let $\psi_w(o_i)$ be the angle over which we have to
rotate~$\rho(w)$ in clockwise direction to make it parallel to~$o_i$; see Figure~\ref{fig:orientation-new}(ii).
Let $\Psi(w) = \{ \psi_1,\psi_2,\ldots \}$, %\as{where $\psi_j < \psi_{j+1}$}\aaa{I think this is implied by sorted.}
for all $1\leq j< |\Psi(w)|$,
be the sorted set of (distinct) angles $\psi_w(o_i)$ defined by the segments in $\cO(w)$. 
For an index~$j$ with $1\leq j\leq |\Psi(w)|$, 
define $\cO(w,j) := \{ o_i \in \cO(w) : \psi_w(o_i) \leq \psi_j\}$,
%\[
%\cO(w,j) := \{ o_i \in \cO(w) : \psi_w(o_i) \leq \psi_j\},
%\]
and define $\cO(w,0)=\varnothing$. 
We call $\cO(w,j)$ a \emph{prefix set}.
The key observation is that $\sol(\Gamma_w)$ must intersect the segments from 
some prefix set $\cO(w,j)$,
where $j$ depends on the unknown vertex~$z$ preceding~$w$.
So our dynamic program will try all possible prefix sets~$\cO(w,j)$,
and make sure that subproblems are combined in a consistent manner.
\medskip

We now have everything in place to describe our dynamic-programming table. It consists of
entries $A[w,j]$, where $w$ ranges over all points in~$Y$, and $j$ ranges over all values
for which% the prefix set
~$\cO(w,j)$ is defined.  %\as{or $O(n^2)$, if $X$ contains
%only endpoints}. 
For convenience add two special entries, $A[\pbot,0]$ and $A[\pbotc,0]$;
the former will serve as the base case, % in the dynamic program, 
and the latter will 
contain (the value of) the final solution. Note that these are the only ones
for $\pbot$ and $\pbotc$, and that we have at most $|Y|\cdot n=O(n^2)$~entries. %Next 
We define the set $\cO^*(w,j)$ of segments to be covered in a subproblem.
%-----------------------------------------------------------------------------------
\begin{quote}
For a point $w\in Y$, the set $\cO^*(w,j)$ consists of the segments~$o_i\in \cO$ 
that satisfy one of the following conditions: 
\begin{quote}
\begin{enumerate}[(i)]
\item $o_i$ intersects the clockwise wedge from $\rho_0$ to $\rho(w)$---note that this 
      wedge need not be convex---but not $\rho(w)$ itself, and $o_i\not\in \cO^-(\rho_0)$; or \label{cond:exact-1}
\item $o_i$ intersects $\pbot w$; or
      \label{cond:exact-2}
\item $o_i\in \cO(w,j)$.
      \label{cond:exact-3}
\end{enumerate}
\end{quote}
Furthermore, $\cO^*(\pbot,0) := \varnothing$ and $\cO^*(\pbotc,0) := \cO$.
\end{quote}
%-----------------------------------------------------------------------------------
We would like now to define $A[w,j]$ to be the minimum length of a convex chain~$\Gamma$ 
from $\pbot$ to $w$ such that all objects in $\cO^*(w,j)$ intersect~$C(\Gamma)$.
There is, however, a technicality to address: the minimum-perimeter polygon that intersects
all segments from $\cO$ need not be convex when we require it to have $\pbot$ as a vertex.
Such a non-convex polygon cannot be the final solution---if the optimum for
a given choice of $\pbot$ is non-convex, then $\pbot$ was not the correct choice---but
it makes a clean definition of our subproblems awkward. Therefore, instead of
first defining the subproblems and then giving the recursive formula, we will immediately
give the recursive formula and then prove that it computes what we want.

For two points $v,w\in Y^+$ (where $Y^+=Y\cup \{\pbotc\}$) with $v\prec w$ and a set $\cO'\subseteq \cO$, let
$L(v,w,\cO')$ be the minimum length of a convex chain~$\Gamma$  from $v$ to $w$ such that
the convex polygon defined by $\Gamma$ and $vw$ is an intersecting set for $\cO'$
and all inner vertices of $\Gamma$ are disjoint from $Y$. Recall that
we can compute $L(v,w,\cO')$ using subroutine~II.
As before, let $\Delta(\pbot,v,w)$ denote the triangle with vertices $\pbot,v,w$. 
%-----------------------------------------------------------------------------------
\begin{definition} \label{def:dp-exact}
Let $w\in Y^+$ and $j$ be a value for which $\cO[w,j]$ is defined.
Thus $0\leq j\leq |\Psi(w)|$, where we set $|\Psi(w)|:=0$ for $w\in\{\pbot,\pbotc\}$.
For $v\prec w$ and $0\leq j'\leq |\Psi(v)|$, let
\[
\cO^*(w,j,v,j') := \cO^*(w,j) \ \setminus \ \Big( \cO^*(v,j') \cup \{ o_i\in \cO : o_i \mbox{ intersects } \Delta(\pbot,v,w) \} \Big)
\]
and define
\[
A[w,j] \ := \ \left\{ \begin{array}{ll}
                    \hspace*{7mm} 0 & \mbox{if $w=\pbot$}  \\
                    % \infty & \mbox{if $v=\vbot$ and not all objects in $\cO(v,w)$ intersect $\vbot w$} \\
                    \min\limits_{\substack{v\prec w \\ 0\leq j'\leq |\Psi(v)|}} L(v,w,\cO^*(w,j,v,j')) + A[v,j'] & \mbox{otherwise.} 
                    \end{array}
             \right.
\]
\end{definition}
%-----------------------------------------------------------------------------------
The next lemma implies that the table entry $A[\pbotc,0]$ defined by this recursive formula 
gives us what we want. Part~(\ref{lem-part-a}) implies that 
$A[\pbotc,0]$ will never return a value that is too small, while part~(\ref{lem-part-b})
implies that for the correct choice of $\pbot$ and range of orientations for
the tangent to $\Popt$ at $\pbot$, the entry $A[\pbotc,0]$ gives us (the value of)
the optimal solution.
%-----------------------------------------------------------------------------------
\begin{lemma} \label{lem:exact-correctness}
Consider the table entry $A[\pbotc,0]$ defined by Definition~\ref{def:dp-exact} for 
a given point $\pbot$ and range $[\psi_i,\psi_{i+1}]$.
\begin{enumerate}[(a)]
\item \label{lem-part-a} There exists a convex intersecting polygon for $\cO$ of perimeter at most $A[\pbotc,0]$.
\item \label{lem-part-b} If $\pbot$ is a vertex of the minimum-perimeter convex intersecting polygon $\Popt$ 
      for $\cO$, and $\rho_0$ does not intersect $\Popt$, and there is a tangent 
      line~$\ell$ at $\pbot$ whose orientation is in the range $[\psi_i,\psi_{i+1}]$, 
      then $\peri(\Popt)=A[\pbotc,0]$.
\end{enumerate}
\end{lemma}
%-----------------------------------------------------------------------------------
\begin{proof}
$\mbox{}$
\begin{enumerate}[(a)]
\item We will prove by induction on $w$ (with respect to the order~$\prec$) that
      for any pair $(w,j)$ for which $A[w,j]$ is defined, there is a sequence 
      $\seq(w,j) = v_0, v_1,\ldots,v_k$ of points from~$Y$, where
      $\pbot = v_0\prec v_1\prec\cdots\prec v_k = w$, with the following property: 
      there are convex polygons $P_1,\ldots,P_k$ such that 
      \begin{itemize}
      \item $v_{i-1}v_i$ is an edge of $P_i$,
      \item $\sum_{i=1}^k \left( \peri(P_i) - |v_{i-1}v_i| \right) = A[w,j]$
      \item $\sol(w,j) := \bigcup_{i=1}^k \left( P_i \cup \Delta(\pbot,v_{i-1},v_i) \right)$ 
            is an intersecting set for $\cO^*(w,j)$. \\[-2mm]
      \end{itemize}
      
      Observe that the boundary length of $\sol(w,j)$ is at most $A[w,j]+|w\pbotc|$.
      Indeed, each edge $v_{i-1}v_i$ is shared between the polygon~$P_i$
      and the triangle~$\Delta(\pbot,v_{i-1},v_i)$, and so the contribution of the
      polygons $P_i$ to the union boundary is at most 
      $\sum_{i=1}^k \left( \peri(P_i) - |v_{i-1}v_i| \right) = A[w,j]$. 
      (At most, because the $P_i$'s are not required to be
      disjoint.) Since the edges $\pbot v_i$ are shared between consecutive triangles 
      $\Delta(\pbot,v_{i-1},v_i)$ and $\Delta(\pbot,v_{i},v_{i+1})$, the contribution of 
      triangles $\Delta(\pbot,v_{i-1},v_i)$ to the union boundary is at most
      $|\pbot v_0| + |\pbot v_k| = |\pbot w|$.
      Thus the boundary length of $\sol(w,j)$ is at most $A[w,j]+|w\pbotc|$, as claimed.
      It follows that the existence of the sequence $\seq(\pbotc,0)$
      implies part~(a) 
      of the lemma, because $\cO^*(\pbotc,0)=\cO$ and so 
      the convex hull of $\sol(\pbotc,0)$ is a convex intersecting set for $\cO$
      of perimeter at most $A[\pbotc,0]$.
      
      The base case is $(w,j)=(\pbot,0)$. Since $\cO^*(\pbot,0)=\varnothing$,
      the sequence consisting of the single point~$\pbot$ trivially
      has the desired properties in this case, with $\sol(\pbot,0)=\{\pbot\}$.
      
      Now consider a pair $(w,j)$ with $\pbot\prec w$. Let $j'$ be such that
      \[
      A[w,j] = L(v,w,\cO^*(w,j,v,j')) + A[v,j'].
      \]
      By induction, there is a sequence $\seq(v,j')$ with
      the properties stated above. We claim that the sequence $\seq(w,j)$
      we obtain by appending $w$ to $\seq(v,j')$ has the desired properties.
      Let $\seq(w,j) = v_0, v_1,\ldots,v_k$ and note that $v_{k-1} = v$ and $v_k=w$.
      Take $P_k:= \sol(\Gamma)$, where~$\Gamma$ is a 
      minimum-length convex chain~$\Gamma$  from $v$ to $w$ such that
      $\sol(\Gamma)$ is an intersecting set for $\cO^*(w,j,v,j')$.
      The induction hypothesis and the fact that $A[w,j] = L(v,w,\cO^*(w,j,v,j')) + A[v,j']$
      now implies that $\seq(w,j)$ has the first two properties.
      The third property follows from the induction hypothesis and
      the definition of~$\cO^*(w,j,v,j')$.
\item Since $A[\pbotc,0]\geq \peri(\Popt)$ by part~(a), 
      it suffices to prove that $A[\pbotc,0]\leq \peri(\Popt)$.
      To this end, let $\pbot=v_0,v_1,\ldots,v_k=\pbotc$ be the points from $Y$
      on $\bd \Popt$, in clockwise order. Let $\Gamma_i$ denote the part
      of $\bd\Popt$ from $v_{i-1}$ to $v_i$ (in clockwise direction),
      and let $\Gamma_{\leq i}$ be the part of $\bd\Popt$ from $\pbot$ to~$v_i$.
      Let $\Popt^{(i)}$ denote the part of $\Popt$ bounded by the segment~$v_{i-1}v_i$
      and $\Gamma_i$; see Figure~\ref{fig:exact-correctness} for an illustration.
      %-----------------------------------------------------------------------------------
      \begin{figure}
      \begin{center}
      \includegraphics{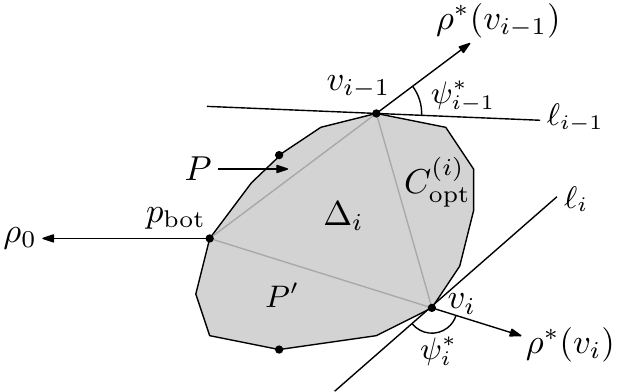}
      \end{center}
      \caption{Illustration for the proof of Lemma~\protect\ref{lem:exact-correctness}.}
      \label{fig:exact-correctness}
      \end{figure}
      %-----------------------------------------------------------------------------------
      
      For $1\leq i<k$, let $\ell_i$ be a line
      tangent to $\Popt$ at~$v_i$, and assume without loss of generality that
      $\ell_i$ is not parallel to any segment in~$\cO$. 
      Recall that $\Psi(v_i)$ denotes the sorted sequence of angles~$\psi_{v_i}(o)$
      defined by the segments in $\cO(v_i)$, that is, by the segments intersecting the
      ray~$\rho^*(v_i)$. Let $\psi_i^*$ be the angle over which we have to rotate
      $\rho(v_i)$ to make it parallel to~$\ell_i$. Finally, let 
      $\Psi(v_i):= \langle \psi^{(i)}_1,\psi^{(i)}_2,\ldots \rangle$ and let
      $j_i$ be the largest index such that $\psi^{(i)}_{j_i} < \psi_i^*$,
      where we define $j_0=j_k=0$.
      \begin{claiminproof}
      $\Popt^{(i)}$ intersects all segments from~$\cO^*(v_i,j_i,v_{i-1},j_{i-1})$.
      \end{claiminproof}
      \begin{proofinproof}
      Partition $\Popt$ into four pieces: $\Popt^{(i)}$, 
      the triangle $\Delta_i := \Delta(\pbot,v_{i-1},v_i)$, and two pieces denoted by~$P$ and $P'$,
      as depicted in Figure~\ref{fig:exact-correctness}. Any segment $o\in \cO$
      intersects at least one of these pieces. We will show that if
      $o$ intersects $\Delta_i$ or~$P$ or~$P'$, then
      $o\not\in  \cO^*(v_i,j_i,v_{i-1},j_{i-1})$, which implies the claim.
      
      Note that for $i=1$ we have $v_{i-1}=\pbot$, in which case $\Delta_i$
      degenerates to a line segment and $P=\varnothing$. Similarly, for $i=k$
      we have $v_{i}=\pbotc$, in which case $\Delta_i$
      degenerates to a line segment and $P'=\varnothing$.
      Finally, when $i=k=1$ we have $\Popt^{(i)}=\Popt$ and the
      claim trivially holds.
      \begin{itemize}
      \item If $o$ intersects $\Delta_i$ then 
            $o\not\in  \cO^*(v_i,j_i,v_{i-1},j_{i-1})$ by definition of 
            $\cO^*(v_i,j_i,v_{i-1},j_{i-1})$.
      \item Now suppose $o$ intersects~$P$ but not $\Delta_i$.
      
            If $o$ intersects $\rho(v_{i-1})$ then it must intersect $\rho^*(v_{i-1})$, 
            since we assumed $o$ does not intersect $\Delta_i$.
            But if $o$ intersects $P$ and $\rho^*(v_{i-1})$, then $\psi_{v_{i-1}}(o) < \psi^*_{i-1}$
            and so $\psi_{v_{i-1}}(o) \leq \psi_{j_{i-1}}^{(i-1)}$. This means that
            $o\in \cO(v_{i-1},j_{i-1}) \subseteq \cO^*(v_{i-1},j_{i-1})$.
            Hence, $o\not\in  \cO^*(v_i,j_i,v_{i-1},j_{i-1})$.
            
            Now assume $o$ does not intersect $\rho(v_{i-1})$.
            Note that $o$ intersects the wedge from $\rho_0$ to $\rho(v_{i-1})$,
            since $o$ intersects $P$. Then the only reason for $o$ to
            not be in $o\in \cO^*(v_{i-1},j_{i-1})$ is if $o\in \cO^-(\rho_0)$.
            But then $o$ cannot intersect $P$.
            Hence, $o\in \cO^*(v_{i-1},j_{i-1})$ and so $o\not\in  \cO^*(v_i,j_i,v_{i-1},j_{i-1})$.
      \item Finally, suppose $o$ intersects~$P'$ but not $\Delta_i$.
            Note that then $v_i\neq \pbotc$, otherwise $P'=\varnothing$.
            
            If $o$ intersects $\rho(v_{i})$ then it must intersect $\rho^*(v_{i})$, 
            since we assumed $o$ does not intersect $\Delta_i$.
            But if $o$ intersects $P'$ and $\rho^*(v_{i})$, then $\psi_{v_i}(o) > \psi^*_i$
            and so $\psi_{v_i}(o) > \psi_{j_i}^{(i)}$. This means that
            $o\not\in \cO(v_{i},j_{i})$. Since $o$ intersects $\rho^*(v_{i})$,
            this means that $o\not\in \cO^*(v_{i},j_{i})$ and
            hence, that $o\not\in  \cO^*(v_i,j_i,v_{i-1},j_{i-1})$.
            
            If $o$ does not intersect $\rho(v_{i})$, then either $o$
            lies entirely in the wedge from $\rho(v_i)$ to $\rho_0$
            or it intersects~$\rho_0$. In the former case $o$ clearly cannot
            satisfy any of the conditions to be in $\cO^*(v_{i},j_{i})$.
            In the latter case the fact that $o$ intersects~$P'$ implies
            that $o\in \cO^-(\rho_0)$, and so
            $o\not\in \cO^*(v_{i},j_{i})$. Hence, $o\not\in  \cO^*(v_i,j_i,v_{i-1},j_{i-1})$.
      \end{itemize}
      \end{proofinproof}
      We now prove by induction on~$i$
      that $A[v_i,j_i]\leq \mylength(\Gamma_{\leq i})$ for all $0\leq i\leq k$. Since 
      $(v_k,j_k)=(\pbotc,0)$ and $\mylength(\Gamma_{\leq k})=\peri(\Popt)$, this will 
      prove part~(b). 
      
      The base case for the induction is 
      when $i=0$. Then we have $A[v_0,j_0] = A[\pbot,0] = 0 = \mylength(\Gamma_0) = \mylength(\Gamma_{\leq 0})$,
      since~$\Gamma_0=\Gamma_{\leq 0}=\{\pbot\}$. Now suppose $i>0$. Then 
      \[
      \begin{array}{lll}
      A[v_i,j_i] & \leq & L(v_{i-1},v_i,\cO^*(v_i,j_i,v_{i-1},j_{i-1})) +  A[v_{i-1},j_{i-1}] \mbox{ (by definition of $A[v_i,j_i]$)}\\
                 & \leq & \mylength(\Gamma_i)+  A[v_{i-1},j_{i-1}] \hfill \mbox{(by the claim above)} \\
                 & \leq & \mylength(\Gamma_i)+  \mylength(\Gamma_{\leq i-1}) \hfill \mbox{(by the induction hypothesis)} \\
                 & = & \mylength(\Gamma_{\leq i}),
      \end{array}
      \]
      which finishes the proof.\qedhere
\end{enumerate}
\end{proof}
%-----------------------------------------------------------------------------------

%-----------------------------------------------------------------------------------
%\mypara{Putting everything together.}
\paragraph*{Putting everything together.}
%-----------------------------------------------------------------------------------
Lemma~\ref{lem:exact-correctness} implies that after solving the dynamic programs
for all choices of $\pbot$ and the range $[\psi_i,\psi_{i+1}]$, we have found
the minimum perimeter intersecting set for $\cO$. (Computing the intersecting
set itself, using the relevant dynamic-program table, is then routine.)
This leads to the proof of Theorem~\ref{thm:2d-exact}.
%-----------------------------------------------------------------------------------
\begin{proof}[Proof of Theorem~\ref{thm:2d-exact}.]
The number of dynamic programs solved is $O(|Y|\cdot n)=O(n^2)$.
The dynamic-programming tables have $O(n^2)$ entries. % \as{$O(|X|\cdot n) = O(n^3)$?}.
Computing an entry takes
$O(|Y|\cdot n)=O(n^2)$ calls to Subroutine~II, at $O(n^3\log n)$ time each.
The dynamic programs thus take $O(n^2) \cdot O(n^2) \cdot O(n^2) \cdot O(n^3\log n) = O(n^{9}\log n)$ time. If the optimal solution does not go through any point of $Y$, then by Theorem~\ref{thm:subroutines} it will be found in $O(n^6\log n)$ time. The optimum of these two algorithms is the global optimum.
%\as{I think the current running time should be $O(n^{12} \log n)$. However if Subroutine~II can work without the intersection points 
%(which I think it can, because of Lemma 21), then the running time would be $O(n^9 \log n)$.
%Based on Lemma 21, we could extend Subroutine~I to work with intersection points as well. Thus in $O(n^6 \log n)$ time
%we could find solutions with no endpoints and in $O(n^9 \log n)$ solutions with at least one endpoint.}
\end{proof}
%-----------------------------------------------------------------------------------

%\input{3d-EPTAS}
%\input{lower-bounds}

\section{Conclusion}

We gave fully polynomial time approximation schemes for the minimum perimeter and minimum area convex intersecting polygon problems for convex polygons. Additionally, we developed a polynomial-time algorithm for the minimum perimeter problem of segments.

It is likely that the running times of our algorithms can be improved further. One could also try to generalize the set of objects, for example, adapting the minimum area algorithm to arbitrary convex objects. We propose the following open questions for further study.
\begin{itemize}
    \item Is there a polynomial-time exact algorithm for the minimum area convex intersecting polygon of segments?
    \item Is there a polynomial-time exact algorithm for minimum perimeter or minimum area convex intersecting polygon of convex polygons, or are these problems \NP-hard?
    \item Is there a polynomial-time approximation scheme for the minimum volume or minimum surface area convex intersecting polytope of convex polytopes in $\Reals^3$? Can we at least approximate the diameter of the optimum solution to these problems?
\end{itemize}
It would be especially interesting to see an \NP-hardness proof for minimum volume or surface area convex intersecting set of convex objects in higher dimensions.

\bibliographystyle{plainurl}
\bibliography{bibliography}

\newpage
\appendix
\section{Visiting rays: on an algorithm of Tan}
\label{app:counterexample}
In this section we give a counterexample to a lemma of Tan~\cite{tan_tour_rays} that is used to establish his results on rays and segments. The paper uses the TSPN framework: given a set of rays in the plane, we want to find the shortest closed curve (a tour) intersecting all the rays.

In~\cite{tan_tour_rays} the concept of a \emph{pseudo-touring rays route} is introduced, which is a tour starting and ending at some fixed point $s$ that is allowed to not visit some rays as long as their supporting line is crossed. The shortest pseudo-tour for a starting point $s$ and ray set $R$ is denoted by POPT$_s(R)$, and the optimum tour visiting the rays and $s$ is denoted by OPT$_s(R)$. Let $\mathcal{H}$ be the convex hull of POPT$_s(R)$ and the starting points of the rays that are not visited by POPT$_s(R)$. We note that \cite{tan_tour_rays}  shows that the  optimal tour and pseudo-tour are unique. We prove the following lemma, which is a contradiction to the statement of Lemma~$6$ in~\cite{tan_tour_rays}:

\begin{lemma} There exists an input ray set $R$ such that the following hold: some ray $r$ is not visited by POPT$_s(R)$, the starting point of $r$ is on $\mathcal{H}$, and  OPT$_s(R)$ does not make a crossing contact with $r$.
\end{lemma}
\begin{proof}
We present an  input instance consisting of a starting point $s=(0,-\epsilon)$, for some $\epsilon<0.1$, and three rays defined as follows (see also Figure~\ref{fig:counterexample}):
\begin{itemize}
    \item $r_1$ has $y=\frac45 x + \frac45$ as supporting line, $o_1 = (-1,0)$ as starting point and the ray is pointing downwards.
    \item $r_2$'s supporting line has the equation $y=1$, starting point $o_2 = (-1,1)$ and the ray is pointing towards the right.
    \item $r_3$'s supporting line has the equation $y=0$, its starting point is $o_3 = (2,0)$ and the ray is pointing towards the right.
\end{itemize}

Let $T_p$ be a pseudo-tour going from $s$ to point $p_1 = (0,1)$ and back to $s$. $T_p$ is a feasible pseudo-tour starting and ending at $s$, since it reflects on $r_2$ and makes crossing contact with the supporting lines of $r_1$ and $r_3$. Furthermore, let $T$ be a tour starting at $s$, visiting $r_1$ at $o_1$, $r_2$ at point $(\frac12,1)$ and $r_3$ at $o_3$ before returning to point $s$. This is a feasible tour starting and ending at $s$ since it visits all three rays of the input instance.

See Figure~\ref{fig:counterexample} for an illustration of the input instance as well as $T_p$ and $T$.

We claim that $T_p$ and $T$ are an optimal pseudo-tour and tour respectively for the input instance. The lemma would then directly follow for $r=r_1$, since $r_1$ is not visited  by $T_p$, its origin $o_1$ is on $\mathcal{H}$ which is defined by the vertices $o_1,p_1,o_3$ and $s$, and $T$ does not make a crossing contact with $r_1$ at $o_1$ (it does what Tan~\cite{tan_tour_rays} calls a bending contact). Since the optimal tour and pseudo tour are unique the lemma then follows.

\begin{figure}[t]
    \centering
    \includegraphics[scale=1.2]{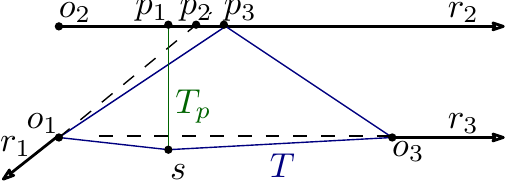}
    \caption{A counterexample to the statement of Lemma~$6$ from~\cite{tan_tour_rays}}
    \label{fig:counterexample}
\end{figure}

Regarding the optimality of $T_p$, note that any feasible pseudo-tour must either reflect on $r_2$ or make a crossing contact with the supporting line of $r_2$. Since the distance between $s$ and the supporting line of $r_2$ is $1+\epsilon$, any feasible pseudo-tour for the instance must therefore have a total length of at least $2+2\epsilon$. Since $T_p$ has a total length of $2+2\epsilon$ it must be optimal.

Regarding the optimality of $T$, we first observe that by the triangle inequality the optimal tour must consist of $4$ line segments with endpoints at $s$, $r_1$, $r_2$ and $r_3$. We first show that $T$ is shortest among all tours that start from $s$, visit $r_1$ followed by $r_2$ and then $r_3$ before returning to $s$. We will then show that any tour visiting the rays in a different order must have a strictly higher cost.

By construction, the shortest possible line segment connecting $s$ and $r_1$ is $so_1$ and similarly the shortest possible line segment connecting $s$ and $r_3$ is $so_3$. Furthermore, the shortest possible path from $r_1$ to $r_3$ through $r_2$ is the one from $o_1$ to $p_3$ to $o_3$. Therefore $T$ is the shortest possible tour that goes from $s$ to $r_1$ to $r_2$ and then $r_3$ before returning to $s$. Note that the length of $T$ is exactly $\sqrt{1+\epsilon^2} + 2\sqrt{1.5^2+1} + \sqrt{4+\epsilon^2}$ which is strictly less than $6.8$ for any $\epsilon<0.1$.

The only two other distinct possible orders of visitation are $s,r_2,r_1,r_3$ and then going back to $s$, and $s,r_1,r_3$ and $r_2$ before returning to $s$. The first one has a total length of at least $(1+\epsilon)+1 + 3 + \sqrt{4+\epsilon^2}$ by taking the respective minimum distances between the objects. This total length is strictly greater than $7$. In a similar way, the second one has a total length of at least $\sqrt{1+\epsilon^2} + 3 + 2\sqrt{2}$ which is strictly greater than $6.8$. This concludes the proof.
\end{proof}

\end{document}